\documentclass[11pt]{article}
\usepackage
[
        a4paper,% other options: a3paper, a5paper, etc
        left=2.5cm,
        right=2.5cm,
        top=3cm,
        bottom=3cm,
]
{geometry}

\usepackage{amsthm}
\usepackage{amsmath}
\usepackage{amssymb}
\usepackage{mathtools}
\usepackage{thmtools}
\usepackage{thm-restate}
\usepackage{algorithm}
\usepackage{algpseudocode}
\usepackage{enumitem}
\usepackage[labelfont=bf]{caption}
\usepackage{import}

\newcommand{\ttt}{\texttt}

\newcommand{\pr}{\mathbb{P}}
\newcommand{\ex}{\mathbf{E}}

\newcommand{\rs}{\mathbb{R}}

\newcommand{\ve}{\varepsilon}
\newcommand{\tx}[1]{\text{#1}}

\newcommand{\poly}{\text{ poly}}

\newcommand{\lp}{\left(}
\newcommand{\rp}{\right)}
\newcommand{\lb}{\left[}
\newcommand{\rb}{\right]}
\newcommand{\bw}{\mathbf{W}}

\newtheorem{theorem}{Theorem}
\newtheorem{lemma}[theorem]{Lemma}

\newtheorem{definition}[theorem]{Definition}

\newtheorem{fact}[theorem]{Fact}

\pdfoutput=1

\begin{document}

\thispagestyle{plain}
\title{An Efficient Semi-Streaming PTAS for Tournament Feedback Arc Set with Few Passes}
\author{Anubhav Baweja\\ CMU \and 
Justin Jia\\ CMU \and David P. Woodruff\\ CMU}
\date{}

\maketitle

%TODO mandatory: add short abstract of the document
\begin{abstract}
We present the first semi-streaming polynomial-time approximation scheme (PTAS) for the minimum feedback arc set problem on directed tournaments in a small number of passes. Namely, we obtain a $(1 + \varepsilon)$-approximation in time $O \left( \text{poly}(n) 2^{\text{poly}(1/\varepsilon)} \right)$, with $p$ passes, in $n^{1+1/p} \cdot \text{poly}\left(\frac{\log n}{\varepsilon}\right)$ space. The only previous algorithm with this pass/space trade-off gave a $3$-approximation (SODA, 2020), and other polynomial-time algorithms which achieved a $(1+\varepsilon)$-approximation did so with quadratic memory or with a linear number of passes. We also present a new time/space trade-off for $1$-pass algorithms that solve the tournament feedback arc set problem. This problem has several applications in machine learning such as creating linear classifiers and doing Bayesian inference. We also provide several additional algorithms and lower bounds for related streaming problems on directed graphs, which is a largely  unexplored territory.  
\end{abstract}

\section{Introduction}

Graph problems have historically been an area of interest because of their various applications, but as the size of these problems becomes very large it becomes essential to design algorithms suitable for models handling such graphs, such as the streaming model. In the streaming model, the input graphs are presented as a read-once tape of edges where a possibly adversarial ordering of observed edges must be accounted for, and low space complexity must be maintained. Much work studying undirected graphs is already present in the streaming literature\cite{feigenbaum2005graph}; we refer the reader to the survey by McGregor \cite{mcgregor2014graph}. However, relatively little has been done, especially on the algorithmic side, in the streaming setting for problems involving directed graphs (digraphs). A notable exception is the the initial investigation conducted by Chakrabarti et al. \cite{chakrabarti_ghosh_mcgregor_vorotnikova_2019} (see also lower bounds in  \cite{assadi2020multi,assadi2020near,chen2021almost,guruswami2016superlinear}), who study the Minimum Feedback Arc Set problem on tournament graphs, among other problems. 

In this work, our focus is also on large directed graphs in the streaming setting, and in particular we present a new algorithm for the Minimum Feedback Arc Set problem on tournament graphs.  A tournament graph is a directed graph where there exists a single edge between all pairs of vertices, and the goal is to find the minimum number of edges that need to be deleted in order to make the graph acyclic. We will also discuss applications and other problems for directed graph streams below.

\subsection{Applications}

The Minimum Feedback Arc Set problem on tournament graphs can be rephrased as the Ranking by Pairwise Comparison (RPC) problem. Given a finite set $V$ and a set of pairwise preference labels (denoted by $u \prec v$ if $v$ is preferred over $u$, where $u, v \in V$), the goal of the RPC problem is to find an ordering of the elements of $V$, from least preferred to most preferred, to minimize the number of disagreements. Finding a suitable global ranking for data, described only by pairwise preference relationships, arises in various practical applications and is commonly referred to as Kemeny-Young Rank Aggregation \cite{kenyon2007rank}. This has applications to machine translation \cite{rosti2007combining} and ranking search engine results \cite{geng2008query}. Whereas the bulk of the learning to rank literature involves rating standalone inputs on a predefined scale, this ranking task involves relative relationships between the objects, whether they are webpage search results or tournament competitors. An ordering of elements that satisfies this condition of minimizing regret (up to some error) and another condition of ``local chaos'' (described by \cite{ailon2012active}), can also be used to create a regularized large margin linear classifier. The main focus of Ailon's work \cite{ailon2012active} was to demonstrate the application's feasibility in the query-efficient setting, and is shown here to also be viable in a semi-streaming setting.

Algorithms for solving the Minimum Feedback Arc Set problem can also be used to decrease the computational cost of  Bayesian inference \cite{festa1999feedback, bar1998approximation} by reducing the weighted loop cutset problem to a weighted blackout-feedback vertex set problem. Bayesian networks are popular among the machine learning community, as they provide an interpretable representation of data, along with varying degrees of conditional independence between attributes. The so-called updating problem in Bayesian networks can be solved using the conditioning method; however, the conditioning method runs in time exponential in the size of a loop cutset, which is potentially large, so reducing the problem and using a feedback arc set approximation would reduce the complexity of solving the updating problem. 

%The clique tree algorithm, also known as the junction tree algorithm, is another important subroutine used in machine learning on Bayesian networks and extracts marginalization by essentially clustering vertices based on cycles. Since the weight of the largest clique is upper bounded by the weight of the union between a loop cutset and its largest parent set, solving the minimum feedback arc set problem can once again provide efficient speedups to desirable queries. \\

% which involves a trail $T$ in the network, which is a subgraph whose underlying graph is just a path, and the discovery of a set $S$ such that any pair of vertices is connected by one trail in the union of $S$ and $T$.

%Furthermore, the Weighted Tournament Feedback Arc Set problem is useful in Rank Aggregation: given a list of rankings of several items, we want to find a ranking that is the most consistent with them based on some objective function. Rank Aggregation has several applications in the real world such as Machine Translation \cite{rosti2007combining} and ranking search engine results \cite{geng2008query}. \\

Other applications in a similar domain for the Minimum Feedback Arc Set problem include collaborative filtering \cite{cohen1998learning}, where ordered recommendations need to be generated for users based on their past preferences and the product history for customers with similar interests. There are also various applications for other directed graph problems such as computing the strongly connected components of a graph, including the study of model checking in formal verification \cite{xie2000implicit} and for data flow analysis in compiler optimization \cite{bhattacharyya1999synthesis}.

%\subsection{Preliminaries}
%We study the multi-pass semi-streaming model. We need the f
%We consider an approximate, multi-pass semi-streaming model, prompting the following notation.
%
%
%\begin{definition}
%\label{definition:poly}

Before describing our results, we need to set up some notation. Let $O^*(f(n))$ denote a function bounded by $f(n) \poly \lp \frac{\log n}{\ve} \rp$.
%Our algorithm to the Feedback Arc Set problem on tournament graphs, presented in the next section, is a polynomial-time $(1+\ve)$-approximation algorithm which uses $O^*(1)$ passes and $O^*(n)$ space in each pass. 
The Feedback Arc Set problem for tournaments is parameterized by $V$, the vertex set, and $E$, the edge set of a tournament graph. One can also use a weight matrix $W$ to represent the edge set $E$ such that $W(u, v) = 1$ if $(u, v) \in E$ and $0$ otherwise. 

\begin{definition}
\label{definition:fascost}
The \textbf{cost of a permutation} $\pi$ on vertices with respect to the vertex set $V$ and weight matrix $W$ is 
$$C(\pi, V, W) = \sum_{u, v \in V : \rho_{\pi}(u) < \rho_{\pi}(v)} W(v, u)$$
where $\rho_{\pi}(u)$ is the rank of vertex $u$ in permutation $\pi$.
\end{definition}
The Feedback Arc Set problem aims to find a permutation $\pi^*$ of vertices for which
$$\pi^* = \text{argmin}_{\pi} C(\pi, V, W)$$

\begin{definition}
\label{definition:fassamplecost}
The \textbf{restricted cost of a permutation} $\pi$ on vertices for a given vertex set $V$ and weight matrix $W$, with respect to the edge set $E$, is 
$$C_E(\pi, V, W) = \frac{\binom{n}{2}}{|E|}\sum_{(u, v) \in E : \rho_{\pi}(u) < \rho_{\pi}(v)} W(v, u)$$
\end{definition}

To later demonstrate lower bounds for other directed graph problems, reductions to communication games are used as in Chakrabati et al. \cite{chakrabarti_ghosh_mcgregor_vorotnikova_2019}, where they proved lower bounds, both in the single pass and the multi-pass setting, for a variety of digraph problems such as performing a topological sort, detecting if a graph is acyclic, as well as finding a minimum feedback arc set. A classic problem in communication complexity is the \ttt{INDEX} problem. Here Alice and Bob are given a vector $x \in \{0,1\}^n$ and an index $i \in [n]$, respectively,  and the goal is for Bob to correctly determine whether $x_i$ is a one or a zero. If only Alice can send a single message to Bob, then the minimum length of this message is $\Omega(n)$ for any randomized protocol which succeeds with probability at least $2/3$. For demonstrating a lower bound for multi-pass algorithms, another useful problem is the set chasing problem \cite{guruswami2016superlinear}, which is discussed in Section \ref{sec:lowerbounds}.

\subsection{Contributions and Previous Work}

\paragraph{Minimum Feedback Arc Set.}
Although the Minimum Feedback Arc Set problem is NP-hard even for tournament graphs \cite{alon2006ranking, charbit2007minimum}, a significant amount of work has been done in order to obtain polynomial time approximations \cite{ailon2008aggregating, coppersmith2006ordering, arora2002new, kenyon2007rank} for the problem on tournament graphs. Furthermore, Chen et al. \cite{chen2021almost} gave an $n^{2- o(1)}$ space lower bound, even for $o(\sqrt{\log n})$ pass streaming algorithms for the feedback arc set problem on general graphs. The maximum acyclic subgraph problem, which is the dual of the minimum feedback arcset problem, also has a lower bound of $O(n^{1 - \ve^{c/p}})$ space for $p$-pass, $(1+\ve)$-approximation polynomial time algorithms, where $c$ is a constant \cite{assadi2020multi}. This further motivates the exploration of algorithms for restricted graphs such as tournament graphs.

Our main contribution (Theorem \ref{theorem:ptas}) is the first algorithm that uses $O(n^{1+1/p} \poly (\log n / \ve))$ space and $p$ passes, while providing a $(1 + \ve)$-approximation, in polynomial time.
Our result gives a significant improvement over \cite{chakrabarti_ghosh_mcgregor_vorotnikova_2019}, which achieved the same pass/space trade-off, but could only give at best a $3$-approximation. Our work also significantly improves other polynomial-time algorithms which achieved a $(1+\ve)$-approximation, as such algorithms either require quadratic memory or a linear number of passes.

Note that given $\log n$ passes, this algorithm achieves $O(n \poly (\log n / \ve))$ space, consistent with the original definition of the semi-streaming model. Previously, Kenyon and Schudy \cite{kenyon2007rank} gave a PTAS for the problem which uses $\Theta(n^2)$ space and Ailon \cite{ailon2012active} introduced an algorithm that specifically reduces the query complexity. The latter has a similar motivation to the streaming model, in that there is limited access to the input, but the model is substantially different from the streaming model. To the best of our knowledge, this is the first work to present such tradeoffs for algorithms that solve the Minimum Feedback Arc Set problem on tournaments up to a $(1+\ve)$-approximation. A summary of previous work in the streaming model on this problem, as well as our results, is given in Table \ref{tab:fast}. 

\paragraph{Additional Streaming Problems on Directed Graphs.}
We also continue the study of other problems on directed graphs in data streams. There has been little research in identifying strongly connected components of a directed graph (digraph) in the streaming model. Laura et al. \cite{laura2011computing} demonstrated an algorithm with $O(n\log n)$ space to find the strongly connected components in the W-stream setting, but we focus on the standard insertion-only framework while using as few passes and memory as possible. Section \ref{sec:other} introduces an algorithm that can be implemented with $\tilde O (n^{1 + 1/p})$ space in $p$ passes, using a deterministic subroutine that finds a Hamiltonian path and is guaranteed to meet the space constraints, in contrast to the commonly adopted KWIKSORT algorithm \cite{chakrabarti_ghosh_mcgregor_vorotnikova_2019} that only does so with high probability guarantees.

Lastly, Section \ref{sec:lowerbounds} demonstrates lower bounds for the space complexity of three directed graph problems: finding strongly connected components, determining if the graph is acyclic, and determining whether there exists a path from a given vertex $s$ to all other vertices. The results are presented for two kinds of space lower bounds: for single pass settings and for multiple pass settings. These strong lower bounds also motivate the shift of emphasis from solving the aforementioned directed graph problems for general inputs to special classes of digraphs, in particular tournament graphs. 

\begin{table}[]
    \centering
    \begin{tabular}{c|c|c|c|c}
         \textbf{Algorithm} & \textbf{Approximation} & \textbf{Time} & \textbf{Passes} & \textbf{Space} \\
         Sort by wins \cite{coppersmith2006ordering} & 5 & Polynomial in $n$ & 1 & $O^*(n)$ \\
         Kwiksort \cite{ailon2008aggregating} & 3 & Polynomial in $n$ & $O^*(1)$ & $O^*(n)$ \\
         Modified Kwiksort \cite{chakrabarti_ghosh_mcgregor_vorotnikova_2019} & 3 & Polynomial in $n$ & $p$ & $O^*(n^{1 + 1/p})$ \\
         PTAS \cite{kenyon2007rank} & $1+\ve$ & Polynomial in $n$ & 1 & $O(n^2)$ \\
         Brute force solve with sketching \cite{chakrabarti_ghosh_mcgregor_vorotnikova_2019} & $1 + \ve$ & Exponential in $n$ & 1 & $O^*(n)$ \\
         Sample and Rank \cite{ailon2012active} & $1+\ve$ & Polynomial in $n$ & $O^*(n)$ & $O^*(n)$ \\
         Our Result & $1 + \ve$ & Polynomial in $n$ & $p$ & $O^*(n^{1 + 1/p})$
    \end{tabular}
    \caption{Algorithms for the Feedback Arc Set problem on Tournament graphs}
    \label{tab:fast}
\end{table}

\subsection{Techniques and Intuition}\label{sec:intuition}
We discuss our key techniques and intuition for our main result for the Feedback Arc Set problem on  tournaments. The techniques employed for our other  results are outlined in Section \ref{sec:other}.
    
Previous work by Kenyon and Schudy \cite{kenyon2007rank} and Ailon \cite{ailon2012active} uses the idea of \textbf{single vertex moves}: given a permutation $\pi$, take a vertex $u$ and move it to an index $j$ such that the cost of the new permutation decreases. Ailon's work shows in order to obtain a $(1+\ve)$-approximation, one does not need to strictly reach a local optimum with respect to cost-improving single vertex moves. Rather, making long moves with significant cost improvement suffices. This holds because the algorithm is recursive, so the shorter single vertex moves can be optimized away in base cases via brute force or the additive approximation algorithm given by Frieze and Kannan \cite{frieze1999quick}, which we can turn into a relative error approximation. Note that this latter algorithm is also used by Kenyon and Schudy \cite{kenyon2007rank}, but we show that this algorithm can also be used in the semi-streaming setting as well.

A brief description of our algorithm is given below:
\begin{enumerate}
    \item The algorithm first obtains an $O(1)$-approximation to the Minimum Feedback Arc Set by sorting the vertices by indegree. This achieves a $5$-approximation as shown by Coppersmith et al. \cite{coppersmith2006ordering}. This is easy to compute in one pass in a stream, and already brings us close to the $(3+\epsilon)$-approximation factor of \cite{chakrabarti_ghosh_mcgregor_vorotnikova_2019}. 
    \item In order to bring the approximation factor down from 5 to $(1+\ve)$, we recursively partition the vertex set into smaller vertex sets, while finding and applying cost-improving single vertex moves to the permutation:
    \begin{itemize}
        \item The base case of the recursion is reached when the input vertex set is small enough that we can brute force through all the possible permutations, and pick the one with the least cost.
        \item Another base case is reached when the cost of the vertex set on the input permutation is quadratic in the size of the vertex set, up to polynomial in $\epsilon$ factors, in which case we use the additive approximation algorithm given in \cite{frieze1999quick}, which we can
        turn into a relative error approximation.
        We are the first to use such an algorithm in the semi-streaming setting.
        \item Otherwise, the algorithm first applies several cost-improving single vertex moves through a method called {\sf ApproxLocalImprove}, partitions the vertex set into two halves, and then recursively optimizes the two subsets. Note that this fixes the relative order of all $(u,v)$ such that $u \in U, v \in V$ in our final output permutation, where $U$ is the left subset and $V$ is the right subset, and the algorithm will make no further effort to optimize this cost which `crosses' between $U$ and $V$. However, applying the single vertex moves before recursing on the subsets ensures that this `crossing' cost is not too large.
    \end{itemize}
    \item {\sf ApproxLocalImprove} is our main subroutine and the method through which we find and apply cost-improving single vertex moves to the input permutation:
    \begin{itemize}
        \item First the algorithm obtains some sample sets from the edge set $E$. These sample sets are denoted $E_{v,i}$, with $v$ being the vertex being moved, and $i$ representing the $i$-th sample set for $v$. The sets $E_{v,i}$ each have size $O(\tx{poly}(\log n / \ve))$ and there are also $O(\tx{poly}(\log n / \ve))$ such sample sets (i.e., the number of different indices $i$) per vertex. A more detailed description of these sample sets is given in Section \ref{sec:formal}. Since the algorithm can only use $O^*(n)$ space, where $n$ is the size of the vertex set, the exact cost improvement due to a single vertex cannot be obtained since the entire vertex set $E$ cannot be stored. Therefore, these sample sets $E_{v,i}$ help approximate the cost improvement of single vertex moves. Moreover, these $O(\tx{poly}(\log n / \ve))$ sample sets per vertex are independent of the permutation; therefore once some single vertex moves have been made, consulting the approximations provided by $E_{v,i}$, the sample sets $E_{v,j}$ (where $j > i$) are still independent of the resulting permutation after applying those single vertex moves. 
        \item Multiple single vertex moves need to be made simultaneously in order to achieve the abovementioned $O^*(1)$ sample sets per vertex. To see this, first note that even when $k$ single vertex moves are made in parallel, the cost improvement achieved by them is similar to the sum of the cost improvements that would have been achieved if $k$ moves were made individually. For 2 single vertex moves this is easy to see since any move affects the cost improvement achieved by another move by an additive factor of $O(1)$. Thus, the cost improvement of doing all those moves in parallel is approximated well by the sum of the individual cost improvements. This is also demonstrated with an example in Figure \ref{fig:svms}. Finally, since the cost of the feedback arc set is bounded by $O(n^2)$, if the total cost improvement achieved by the parallel moves is quadratic in $n$ (up to $\log n$ and $\ve$ factors), then the algorithm would only need to make these parallel moves $O^*(1)$ times, justifying the number of sample sets used and showing that the number of passes required for this {\sf ApproxLocalImprove} subroutine is also $O^*(1)$.
\end{itemize}
\end{enumerate}

\begin{figure}[ht]
    \centering
    \includegraphics[width=10cm]{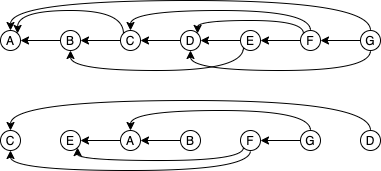}\\
    \caption{The image on top represents the tournament graph before the single vertex moves $(C \to 1), (E \to 2), (D \to 7)$ ($(E \to 2)$ represents moving vertex $E$ to index 2) are made, and the image on the bottom is the graph obtained after making all the moves in parallel (the graph only displays back edges: any edges not displayed go from left to right). The 3 moves improve the cost by 2, 1, 3 respectively, and the total cost is improved by 5 after making all the moves simultaneously even though the moves overlap.}
    \label{fig:svms}
\end{figure}
    
The above helps achieve $O^*(n)$ space and $\log n$ passes, but in order to obtain the desired result of $O^*(n^{1 + 1/p})$ space and $p$ passes, multiple passes of the algorithm are emulated in one pass by increasing the size of the sample sets for all vertices. Specifically, the sample sets for the largest nodes in a layer $L$ are obtained first, which are also used to approximate the cost improvement of single vertex moves in the smaller nodes of $L$. This way, nodes in different layers of the recursion tree can all be handled in the same pass.

\section{Feedback Arc Set for Tournament Graphs}\label{sec:feedback}
Our main algorithm departs from the methods of the earlier streaming algorithm of \cite{chakrabarti_ghosh_mcgregor_vorotnikova_2019},  and is instead inspired by the algorithm of  \cite{ailon2012active} for the query model. However, in order to obtain a streaming algorithm, we need to make a number of crucial modifications to the algorithm of \cite{ailon2012active}. First, our algorithm executes many single vertex moves in parallel, see lines 5-26 in {\sf ApproxLocalImprove} below. Second, we define a notion of {\bf realized improvement in cost} to deal with the effects of multiple vertex moves being made in parallel. Third, we need to balance the space complexity of our algorithm by running an additive error algorithm in parallel; see the {\sf AddApproxMFAS} algorithm below, which we convert to relative error in our context. 

%utilizes several ideas from Ailon's algorithm \cite{ailon2012active}, but makes several crucial observations to reduce the number of passes. It mainly differs from Ailon's algorithm in lines 5-26 of {\sf ApproxLocalImprove} described below. Instead of executing one single vertex move at a time, multiple single vertex moves are performed at once and the samples $E_{v,i}$ are refreshed with respect to all of the moves simultaneously. 

\subsection{Definitions and Background}\label{ssec:definitions}

\begin{definition}
\label{definition:svm}
A \textbf{single vertex move} $(u \to j)$ on permutation $\pi$ is defined as moving the vertex $u$ in $\pi$ to rank $j$, and shifting all the vertices from their original position to their new position accordingly. The resulting permutation is denoted $\pi_{u \to j}$. Additionally, if $M$ is a set of single vertex moves, then $\pi_M$ represents the permutation obtained after making all moves in parallel on permutation $\pi$. 
\end{definition}

Note that the single vertex move described above is not a vertex swap. Additionally, for a set $M$ of moves, if $(u_1 \to j_1) \in M$ and $(u_2 \to j_2) \in M$ then it must be true that $u_1 \neq u_2$, otherwise the moves cannot be made simultaneously. If multiple vertices were to be moved to the same position, this conflict would be resolved by preserving their original relative ordering.

\begin{definition} 
\label{definition:testmove}
$\text{TestMove}(\pi, V, W, u, j)$ is the improvement in cost of the feedback arc set induced by a permutation $\pi$, when  the single vertex move $u \to j$ is conducted to obtain the permutation $\pi_{u \to j}$. That is, if without loss of generality $j \geq \rho_{\pi}(u)$, then
$$\text{TestMove}(\pi, V, W, u, j) = \sum_{v:\rho_{\pi}(v) \in [\rho_{\pi}(u) + 1, j]} \big(W(v,u) - W(u,v)\big)$$
\end{definition}

Note that $\text{TestMove}(\pi, V, W, u, j)$ can be computed easily for arbitrary inputs $\pi, u, j$ if $\Theta(n^2)$ space were allowed. However, approximations become necessary to achieve less memory. The following definition provides an unbiased estimator for $\text{TestMove}(\pi, V, W, u, j)$ without observing all the edges of the graph.

\begin{definition}
\label{definition:testmoveapprox}
$\text{TestMove}_E(\pi, V, W, u, j)$ is an approximation to the improvement in cost of the feedback arc set from $\pi$ to $\pi_{u \to j}$, obtained by only considering the edges in the subset of edges $\tilde{E} = \{(v, u) \in E : \rho_{\pi}(v) \in [\rho_{\pi}(u) + 1, j]\}$. Assuming without loss of generality that $j \geq \rho_{\pi}(u)$, then
$$\text{TestMove}_E(\pi, V, W, u, j) = \frac{j - \rho_{\pi}(u)}{|\tilde{E}|} \sum_{v:(v,u) \in \tilde{E}} \big(W(v,u) - W(u,v)\big)$$
If $E \in \binom{V}{2}$ is chosen uniformly at random from from all multisets of a given size, $\text{TestMove}_E(\pi, V, W, u, j)$ acts as an empirical unbiased estimator of $\text{TestMove}(\pi, V, W, u, j)$ \cite{ailon2012active}.
\end{definition}

As will be seen later, multiple single vertex moves need to be made with respect to the same sample set to obtain the desired space bounds. Making these moves simultaneously affects the cost of other moves being made, which necessitates the notion of \textbf{realized cost improvement} for each move. In other words, there needs to be some quantification of how much worse off the cost would be if that move were not made, and this is made explicit by the next definition.

\begin{definition}
\label{definition:testmoveparallel}
Let $M$ be a set of single vertex moves. $\text{TestMove}^M(\pi, V, W, u, j)$ is the \textbf{realized improvement in cost} of the feedback arc set due to the single move $(u \to j)$ within $M$. That is,
$$\text{TestMove}^M(\pi, V, W, u, j) = C(\pi_M, V, W) - C(\pi_{M \setminus \{(u \to j)\}}, V, W)$$
\end{definition}

Other than the above definitions which will be used in the description and analysis of our main algorithm, there is another result that will be useful:

\begin{definition}
\label{definition:mas}
Given a directed graph, the goal of the \textbf{maximum acyclic subgraph} (MAS) problem is to output an ordering $\pi$ on the vertices $V$ such that the number of forward edges is maximized (the forward edges making up the acyclic subgraph). This is identical to the MFAS problem since the goal of that problem is to output a permutation such that it minimizes the number of backward edges.
\end{definition}

\begin{fact}
\label{fact:addapprox}
\cite{arora2002new}, \cite{frieze1999quick}
There exists a randomized algorithm polynomial-time approximation scheme for the MAS problem on tournament graphs. Given $\beta$, $\eta > 0$, the algorithm {\sf AddApproxMAS}$(V, W, \beta, \eta)$ outputs, in time $O \lp n^2+2^{O(1/\beta^2)} \log \frac{1}{\eta} \rp$ in \cite{frieze1999quick}), an ordering $\pi$ on $V$ whose cost with respect to $W$ is at least $OPT - \beta n^2$, with probability at least $1 - \eta$, where $OPT$ is the size of the largest acyclic subgraph.
\end{fact}

Solving the maximum acyclic subgraph problem and minimum feedback arc set problem on a tournament graph produces the same ordering $\pi$. To see this let $A, B$ be the size of the minimum feedback arc set and maximum acyclic subgraph, respectively, due to the ordering $\pi$. Then $A + B = \binom{|V|}{2}$. Clearly, minimizing $A$ maximizes $B$. Therefore the above fact can be rephrased to solve the minimum feedback arc set problem on tournament graphs, that is, there exists an algorithm {\sf AddApproxMFAS}$(V, W, \beta, \eta)$ with all of the properties mentioned above which outputs an ordering $\pi$ such that its cost is at most $OPT + \beta n^2$ with probability
at least $1 - \eta$.

As will be seen later, our proposed algorithm uses the {\sf AddApproxMFAS} algorithm as a subroutine with parameters $\beta = \ve^3$ and $\eta = \frac{1}{n^4}$:

\begin{theorem}
\label{thm:addapprox}
Given $\beta = \ve^3$ and $\eta = \frac{1}{n^4}$, {\sf AddApproxMAS}$(V, W, \beta, \eta)$ can output an ordering with cost at least $OPT - \ve^3 n^2$, with probability at least $1 - \frac{1}{n^4}$, in time $O \lp n^2+2^{O(\ve^{-6})} \log n \rp$, using $O^*(n)$ space and $O^*(1)$ passes.
\end{theorem}

A brief description of the {\sf AddApproxMAS} algorithm is given in Appendix \ref{sec:addapprox-semi}, as well as the proof for the above theorem.\\

Finally, we introduce the main result of \cite{ailon2012active}, which is used as a subroutine in part of our algorithm. To describe the result we first need the following definition:

\begin{definition}
\label{definition:epsgoodpartition}
Given a set $V$ of size $n$, an ordered decomposition is a list of pairwise disjoint subsets $V_1, ..., V_k \subseteq V$ such that $\bigcup_{i=1}^k V_i = V$. If $\Pi(V)$ is the set of all permutations on $V$, we say that $\pi \in \Pi(V)$ respects $V_1, ..., V_k$ if for all $u \in V_i, v \in V_j, i < j,$ we have $\rho_{\pi}(u) < \rho_{\pi}(v)$ We denote the set of
permutations $\pi \in \Pi(V)$ respecting the decomposition $V_1, ..., V_k$ by $\Pi(V_1, ..., V_k)$. A decomposition $V_1, ..., V_k$ is \textbf{$\ve$-good} with respect to $W$ if 
$$\min_{\pi \in \Pi(V_1, ..., V_k)} C(\pi, V, W) \leq (1+\ve) \min_{\pi \in \Pi(V)} C(\pi, V, W)$$
\end{definition}

\begin{fact}
\label{fact:ailon}
\cite{ailon2012active}
Given a vertex set $V$, a weight matrix $W$ (describing the edges of the graph), and an error tolerance parameter $0 < \ve < 1$, there exists a polynomial time algorithm which returns, with constant probability, an $\ve$-good partition of $V$, querying at most $O(\ve^{-6}n \log^5 n)$ locations in W in expectation. The running time of the algorithm is $O(n \tx{ poly}(\log n, \ve^{-1}))$. This algorithm is referred to as {\sf SampleAndRank}.
\end{fact}

Note that the above result does not output a permutation $\pi$ with near-optimal cost: it outputs a partition $(V_1, ..., V_k)$ of $V$ such that the optimal permutation that respects the partition is near-optimal globally. As it turns out, it is easy to obtain near-optimal permutations on the individual $V_i$ because
\begin{enumerate}
    \item They are small enough to enumerate over all permutations, or
    \item Their optimal cost is high enough that {\sf AddApproxMFAS} outputs a good approximation.
\end{enumerate}

We now describe the algorithm in detail.

\newpage

\subsection{Formal Description of Algorithm}
\label{sec:formal}
We now state our algorithms and subroutines, providing a concise English description and pseudocode for each. Afterwards, we give our formal analysis. 

\begin{algorithm}[H]
\caption{{ \sf GetNearOptimalPermutation}$(V, W, \ve)$: Top level function for computing our $(1+\ve)$-approximation to the minimum feedback arc set. It obtains an $O(1)$-approximation to the minimum feedback arc set, and then calls {\sf Recurse}. The latter recursively executes single vertex moves to bring down the approximation factor to $(1+\ve)$.}
\begin{algorithmic}[1]
\State $n \leftarrow |V|$
\State $\pi \leftarrow O(1)$-approximation by sorting by indegree as described by Coppersmith et al. \cite{coppersmith2006ordering}
\State \Return {\sf Recurse}$(V, W, \ve, n, \pi)$
\end{algorithmic}
\label{algorithm:one}
\end{algorithm}

\begin{algorithm}[H]
\caption{{\sf Recurse}$(V, W, \ve, n, \pi)$: Recursive part of the algorithm that optimizes the current vertex set and then recursively optimizes partitions of the vertex set. The base case of the recursion is reached when the vertex set is small enough to enumerate over all possible $N!$ permutations. The other base case is reached when the cost of the vertex set is quadratic in $N$, which is when the {\sf AddApproxMFAS} algorithm will be effective. If a base case is not reached, then the algorithm calls {\sf ApproxLocalImprove} to perform cost-improving single vertex moves before recursing on the two subsets of $V$. Note that this procedure is similar to the {\sf RecurseSAR} procedure given by Ailon \cite{ailon2012active}, except in the base cases their algorithm returns the trivial partition $\{V\}$ since they only need to output an $\ve$-good partition. }
\begin{algorithmic}[1]
\State $N \leftarrow |V|$
\If {$N \leq \log n / \log \log n$}
\State $\pi^* \leftarrow $ Minimizer for $V, W$ obtained by brute force
\State \Return{$\pi^*$}
\EndIf
\State $E \leftarrow$ random subset of $O(\ve^{-4} \log n)$ elements from $\binom{V}{2}$ (with repetition)
\State $C \leftarrow C_E(\pi, V, W)$ \quad  ($C$ is an additive $O(\ve^2N^2)$ approximation of $C(\pi, V, W)$)
\If {$C = \Omega(\ve^2N^2)$}
\State \Return{{\sf AddApproxMFAS}$(V, W, \ve^3, 1/n^4)$}
\EndIf
\State $\pi_1 \leftarrow$ ApproxLocalImprove$(V, W, \ve, n, \pi)$ 
\State $k \leftarrow$ uniformly random integer in the range $[N/3, 2N/3]$
\State $V_L \leftarrow \{ v \in V : \rho_{\pi}(v) \leq k \}, \pi_L \leftarrow $ restrict $\pi_1$ to $V_L$
\State $V_R \leftarrow \{ v \in V : \rho_{\pi}(v) > k \}, \pi_R \leftarrow $ restrict $\pi_1$ to $V_R$
\State \Return{Concatenate {\sf Recurse}$(V_L, W, \ve, n, \pi_L)$ and {\sf Recurse}$(V_R, W, \ve, n, \pi_R)$}
\end{algorithmic}
\label{algorithm:two}
\end{algorithm}

% \begin{algorithm}[H]
% \caption{GetSample$(\pi, E, u, j)$: Helper function used by {\sf ApproxLocalImprove} to obtain the relevant samples for the single vertex move $(u \to j)$ from the sample $E$. The function filters out the edges in $E$ which }
% \begin{algorithmic}[1]
% \If{$\rho_{\pi}(u) < j$}
% \State {$\hat{E} \leftarrow \{(u, v) : \rho_{\pi}(v) \in [\rho_{\pi}(u)+1, j]\}$}
% \Else
% \State{$\hat{E} \leftarrow \{(u, v) : \rho_{\pi}(v) \in [j, \rho_{\pi}(u)-1]\}$}
% \EndIf
% \State \Return{$E \cap \hat{E}$}
% \end{algorithmic}
% \label{algorithm:three}
% \end{algorithm}

\begin{algorithm}[H]
\caption{{\sf GetMoves}$(V, W, \ve, n, \pi, \{E_{v,i} : v \in V \}, c)$: Helper function used by {\sf ApproxLocalImprove} to obtain cost-improving single vertex moves on $\pi$. The log length of the single vertex move ($\log \lceil |j - \rho_{\pi}(u)|\rceil$) must lie between $B$ and $L$, which are both logs of terms linear in $N$. Moreover, the cost improvement as approximated by the sample sets $\{E_{v,i} : v \in V \}$ must be at least $c |j - \rho_{\pi}(u)| \ve / \log n$.}
\begin{algorithmic}[1]
\State $B \leftarrow \lceil \log (\Theta(\ve N / \log n)) \rceil, L \leftarrow \lceil \log N \rceil$
\State \Return{$\{(u \to j) : u \in V$ and $j \in [n]$ and \qquad \qquad [using $d = |j - \rho_{\pi}(u)|$ and $l = \lceil \log d \rceil$]\\
\qquad \qquad $l \in [B, L]$ and TestMove$_{E_{u,i}}(\pi, V, W, u, j) > c d \ve / \log n$ \}}
\end{algorithmic}
\label{algorithm:four}
\end{algorithm}

\begin{algorithm}[H]
\caption{{\sf ApproxLocalImprove}$(V, W, \ve, n, \pi)$: Function to repeatedly apply cost-improving single vertex moves. This optimizes the permutation on $V$ as long as there exists a single vertex move with linear cost improvement (up to $\log n$ and $\ve$ factors). No cost-improving single vertex moves need to be made if $V$ is small enough. Otherwise, on lines 5-13 the algorithm first obtains sample sets $\{E_{v,i}\}$, where $\{E_{v,i} : v \in V\}$ and $\{E_{v,j} : v \in V\}$ are independent and identically distributed for all $i \neq j$. Once the samples are obtained, in lines 15-25 they are used to perform significant cost-improving single vertex moves until no such moves exist. Section \ref{section:correct} discusses why the algorithm succeeds in doing so with only poly$(\log n / \ve)$ iterations.}
\begin{algorithmic}[1]
\State $N \leftarrow |V|, B \leftarrow \lceil \log (\Theta(\ve N / \log n)) \rceil, L \leftarrow \lceil \log N \rceil$
\If {$N = O(\ve^{-3}\log^3 n)$}
\State \Return{$\pi$}
\EndIf
\For{$v \in V$}
\For {$i \in [1, \Theta(\ve^{-6} \log^6 n)]$}
\State $E_{v,i} \leftarrow \phi$
\For {$m \in [1, \Theta(\ve^{-5} \log^5 n)]$}
\State $j \leftarrow$ uniformly random integer chosen from $[1, N]$
\State $E_{v,i} \leftarrow E_{v,i} \cup \{(v, \pi(j)\}$
\EndFor
\EndFor
\EndFor
\State $i \leftarrow 1$
\While {$|S| > 0$ such that $S =$ {\sf GetMoves}$(V, W, \ve, n, \pi, \{E_{v,i}\}, 1)$}
\State $M_1 \leftarrow $ {\sf GetMoves}$(V, W, \ve, n, \pi, \{E_{v,i}\}, \frac12)$
\For {$m \in [1, \Theta(\ve^{-2} \log^2 n)]$}
\State $M_2 \leftarrow \{(u \to j) \in M_1: l \in [B, L] \tx{  and  }$ $\tx{TestMove}_{E_{u,i}}(\pi, V, W, u, j) > \frac12 \ve d / \log n \}$ \\
\qquad \qquad \qquad [using $d = |j - \rho_{\pi}(u)|$, $l = \lceil \log d \rceil$, \\
\qquad \qquad \qquad $d$ is the 'length' of the single vertex move $(u \to j)$ on $\pi$]
\State $M_3 \leftarrow $ Sample $\Theta(\ve^2 \log^{-2} n)$-fraction of $M_2$ (All $u$'s must be unique, otherwise cannot move simultaneously)
\State $\pi \leftarrow \pi_{M_3}$
\State $i \leftarrow i + 1$ \quad ($E_{u,i}$ is no longer independent of $\pi$, so need a fresh set of samples)
\EndFor
\EndWhile
\State \Return{$\pi$}
\end{algorithmic}
\label{algorithm:five}
\end{algorithm}

Our main algorithm along with the required subroutines is shown above, and our corresponding theorem is formalized below.\\

\begin{theorem}
\label{theorem:ptas}
The algorithm {\sf GetNearOptimalPermutation}:
\begin{enumerate}[label=\roman*.]
    \item returns a $(1+\ve)$-approximation to the Minimum Feedback Arc Set problem on tournaments with constant probability in time $O \left( \text{poly}(n) 2^{\text{poly}(1/\varepsilon)} \right)$, 
    \item requires at most $O^*(n)$ space if executed in $O(\log n)$ passes, and 
    \item can be executed in $p$ passes and $O^*(n^{1 + 1/p})$ space.
\end{enumerate}
\end{theorem} 

The top level algorithm {\sf GetNearOptimalPermutation} first obtains a $5$-approximation by sorting by indegree, which can be achieved by storing the indegree of all vertices \cite{coppersmith2006ordering}. It then calls {\sf Recurse}. The base case occurs when the vertex set becomes small enough to use brute force, or when the cost of the minimum feedback arc set becomes large enough to approximate with {\sf AddApproxMFAS} (which is the source of the $O \left( 2^{\text{poly}(1/\varepsilon)} \right)$ factor in the time). Given this recursion tree and the output permutation $\pi^O$, we can divide the cost of the output permutation into two sources:
\begin{enumerate}
    \item Cost incurred at the internal nodes $\mathcal{I}$: Let $X \in \mathcal{I}$ be an internal node of the recursion tree, and let $L, R$ be the left and right child of $X$. Once {\sf Recurse} has divided $V_X$ into $V_L$ and $V_R$, for any $u \in V_L, v \in V_R$, it will be the case that $\rho_{\pi^O}(u) < \rho_{\pi^O}(v)$ where $\pi$ is the output permutation: their relative ordering is now fixed and will never be changed. If $W(v,u) = 1$, then this is a permanent cost that will be attributed to the parent $X$, which we can quantify as 
    $$\beta_X = \sum_{u \in V_L, v \in V_R} W(v,u)$$
    \item Cost incurred at the leaves $\mathcal{L}$: Let $X \in \mathcal{L}$ be a leaf of the recursion tree. The cost incurred by the leaf is given by
    $$\alpha_X = \sum_{u,v \in V_X : \rho_{\pi^O}(u) < \rho_{\pi^O}(v)} W(v,u)$$
\end{enumerate}
Using the above notation, it is easy to see that $C(\pi^O, V, W) = \sum_{X \in \mathcal{I}} \beta_X + \sum_{X \in \mathcal{L}} \alpha_X$.

Note that {\sf SampleAndRank} \cite{ailon2012active} is similar to the {\sf GetNearOptimalPermutation}: the top level function obtains an $O(1)$ approximation, and then cost-improving single vertex moves are applied recursively. The structure of recursion in {\sf SampleAndRank} is identical to the one in {\sf Recurse}, except {\sf SampleAndRank} returns the trivial partition $\{ V \}$ in the base cases, whereas {\sf Recurse} optimizes the cost of the base cases using either brute force or {\sf AddApproxMFAS}. However, the procedure used by the two algorithms to find cost-improving single vertex moves ({\sf ApproxLocalImprove}) is very different. A detailed description of the {\sf SampleAndRank} algorithm is given in Appendix \ref{sec:sar}.

One difference that stands out in particular between the two local optimization procedures is how the sample sets $E_{v,i}$ are defined. In the lines 5-13 of {\sf ApproxLocalImprove}, a sample set $E_{v,i}$ consists of $\Theta(\ve^{-5} \log^5 n)$ edges that start at $v$ and end at a uniformly random vertex in the input set $V$. If $i \neq j$, then $E_{v,i}$ and $E_{v,j}$ are independent and identically distributed random sets. This particular sampling scheme lets the algorithm get away with $O^*(1)$ number of passes and $O^*(1)$ space within one call of {\sf ApproxLocalImprove}.  

Intuitively, given the above similarities between {\sf SampleAndRank} and {\sf GetNearOptimalPermutation}, any statement that is proved by \cite{ailon2012active} for {\sf SampleAndRank} about $\beta_X$ for $X \in \mathcal{I}$, should hold for {\sf GetNearOptimalPermutation} as well, as long as {\sf ApproxLocalImprove} makes all single vertex moves that {\sf SampleAndRank}'s local optimization procedure would make. The key lemma in the query model from \cite{ailon2012active} that is useful is the following:
\begin{lemma}
\label{lemma:beta}
\cite{ailon2012active}
Let $C^*$ be the cost of the optimal permutation, and $C^*_X$ be the optimal permutation when $V, W$ are restricted to the vertices in $X$. Then {\sf SampleAndRank} guarantees
$$E \lb \sum_{X \in \mathcal{I}} \beta_X \rb \leq (1+\ve)C^* - E \lb \sum_{X \in \mathcal{L}} C^*_X \rb$$
where the expectation is taken over the choice of the pivots $k$ (line 12, Algorithm \ref{algorithm:two}: {\sf Recurse}).
\end{lemma}

Again, in order to use the above lemma, an equivalence between the terminating condition of {\sf ApproxLocalImprove} and the terminating condition of {\sf SampleAndRank}'s local optimization procedure must be shown: if {\sf ApproxLocalImprove} outputs $\pi$, then there cannot exist a \textbf{strong single vertex move} on $\pi$ anymore. A strong single vertex move $(u \to j)$ has the following two properties:
\begin{itemize}
    \item It is long : $|j - \rho_{\pi}(u)| \geq \Theta(\ve N / \log n)$, and
    \item It is effective : $\text{TestMove}(\pi, V, W, u, j) \geq \Theta(\ve |j - \rho_{\pi}(u)| / \log n)$.
\end{itemize}
As long as there is no single vertex move satisfying both the above properties, the algorithm would obey the terminating condition of {\sf SampleAndRank}'s local optimization procedure. Proving this statement nearly establishes the first part of Theorem \ref{theorem:ptas}, and that is the key objective of the following section.

\subsection{Proof of Correctness}
\label{section:correct}

Let $d = |j - \rho_{\pi}(u)|$ and $l = \lceil \log d \rceil$ for the single vertex move $(u \to j)$ on permutation $\pi$, and let $\alpha$ be an arbitrary integer in the range $[1, \Theta(\ve^{-6} \log^6 n)]$. The lemmas shown below outline the structure of the justification for the first part of Theorem \ref{theorem:ptas}, along with their proofs. The following is a brief overview of the proof:

\begin{itemize}
    \item When moves are made in parallel, and approximated using the samples $E_{v,i}$ (generated in lines 5-13 of Algorithm \ref{algorithm:five} : {\sf ApproxLocalImprove}), we want to ensure that the realized cost of a single move is at least $\Theta(\ve^2 N / \log^2 n)$. Lemmas \ref{lemma:ailon1}, \ref{lemma:elementsinmove}, \ref{lemma:testmoveapproxerror}, \ref{lemma:testmoveapproxerrorall} show that the samples approximate the cost improvement of these moves well and all approximations hold with constant probability. Lemmas \ref{lemma:testmoveparallel} and \ref{lemma:realizedimprovement} show that strong single vertex moves have a high realized cost improvement, even if they are made in parallel and approximated using samples.
    \item On a vertex set of size $N$, the maximum feedback arc set size can be $O(N^2)$. If a move is strong, its realized cost improvement is $\Theta(\ve^2 N / \log^2 n)$ as we saw. Now, it is possible that by making all these moves, we create the opportunity for a new strong single vertex move that did not exist earlier. However, lemmas \ref{lemma:enoughsvms} and \ref{lemma:enoughsvms-fascost} show that the overall cost improvement needed to create these new strong single vertex moves is $\Omega(\ve^4 N^2 / \log^4 n)$. 
    \item Given that the maximum cost of a feedback arc set is $\Theta(N^2)$, these new strong moves can only be created $O(\ve^{-4} \log^4 n)$ times. Therefore, as long as the loop on line 15 in {\sf ApproxLocalImprove}, there will be no strong single vertex moves left, and the local optimization can safely end. This is shown in lemmas \ref{lemma:line15approxlocalimprove} and \ref{lemma:line14approxlocalimprove}. 
    \item Finally, in the next section it is formally shown why satisfying the terminating condition of ``no strong single vertex moves remaining'' is sufficient to finish the proof. 
\end{itemize}

We can now delve into the details of this proof.

\begin{lemma}
\label{lemma:ailon1}
[Ailon \cite{ailon2012active}] Let $E \subseteq \binom{V}{2}$ be a random multi-set of size $m$ with elements $(u, v_1), \ldots, (u, v_m)$ such that for each vertex $v_i$, its rank $\rho_{\pi}(v_i)$ is between $\rho_{\pi}(u)$ (exclusive) and $j$ (inclusive) and is chosen uniformly at random in this range. Then  
$$\ex[\tx{TestMove}_E(\pi, V, W, u, j)] = \tx{TestMove}(\pi, V, W, u, j)$$
and additionally, for any $\delta > 0$, with probability of failure $\delta$ it holds that 
$$|\tx{TestMove}_E(\pi, V, W, u, j) - \tx{TestMove}(\pi, V, W, u, j)| = O \lp d \sqrt{\frac{\log \delta^{-1}}{m}}\rp$$
\end{lemma} 
\begin{proof}

This follows by a Hoeffding bound and the fact that $|W(u,v)| \leq 1$ for all $u,v$.
\end{proof}

\begin{lemma}
\label{lemma:elementsinmove}
Given single vertex move $(u \to j)$ on permutation $\pi$ and sample set $E_{u,\alpha}$, \\$\tx{GetSample}(\pi, E_{u,\alpha}, u, j)$ has at least $\Omega(\ve^{-4} \log^4 n)$ elements that are drawn uniformly at random between $\rho_{\pi}(u)$ (exclusive) and $j$ (inclusive), with $1 - O(n^{-6})$ probability of success.
\end{lemma}
\begin{proof}

Let $Y = |\tx{GetSample}(\pi, E_{u,\alpha}, u, j)|$. Let $Y_i$ be an indicator such that $Y_i = 1$ iff the $i$-th sample $(u, v_i) \in E_{u,\alpha}$ is such that $v_i$ lies in the range of the single vertex move $(u \to j)$. Therefore $Y = \sum Y_i$, and $\ex[Y] = \Theta \lp N^{-1} \ve^{-5} d \log^5 n \rp \geq \Theta \lp \ve^{-4} \log^4 n \rp$. Using a Chernoff bound for binomial random variables we get $$\pr(Y \leq \Theta(\ve^{-4}\log^4 n)) \leq e^{-\Theta(\ve^{-4}\log^4 n)} \leq \Theta(n^{-6})$$
as desired. Note that since all elements were picked uniformly at random with repetition in $E_{u,\alpha}$, they are still uniformly random in the required range.
\end{proof}

The above two lemmas are useful for proving that these samples $E_{v,\alpha}$ can be properly utilized to approximate the cost improvement of single vertex moves, as is shown by the next lemma.

\begin{lemma}
\label{lemma:testmoveapproxerror}
Let $\hat{E} = \tx{GetSample}(\pi, E_{u,i}, u, j)$. With probability of success $1 - O(n^{-6})$,$$|\tx{TestMove}_{\hat{E}}(\pi, V, W, u, j) - \tx{TestMove}(\pi, V, W, u, j)| \leq \frac18 d \ve / \log n$$
\end{lemma}
\begin{proof}

By Lemma \ref{lemma:elementsinmove} we have $|\hat{E}| = \Omega(\ve^{-4} \log^4 n)$ with at least $1 - O(n^{-6})$ probability. Setting the probability of failure $\delta = e^{-\ve^{-2}\log^2 n}$, by Lemma \ref{lemma:ailon1} we get
$$|\tx{TestMove}_{\hat{E}}(\pi, V, W, u, j) - \tx{TestMove}(\pi, V, W, u, j)| = O \lp d \sqrt{\frac{\ve^{-2}\log^2 n}{\ve^{-4} \log^4 n}} \rp = O(d\ve/\log n)$$
Note that $\delta$ can be upper bounded by $O(n^{-6})$. Therefore the claim is true with probability at least $1 - O(n^{-6})$.
\end{proof}

The high probability guarantee of a decent approximation for a single vertex move must be extended to all single vertex moves, across all nodes, with a constant probability, in order to prove that the overall algorithm works with constant probability (say $0.99$).

\begin{lemma}
\label{lemma:testmoveapproxerrorall}
All sampling approximations of TestMove in Lemma \ref{lemma:testmoveapproxerror} succeed simultaneously with constant probability.
\end{lemma}
\begin{proof}

There are $O(n)$ many calls to {\sf ApproxLocalImprove} with high probability. Since the loop on line 15 in {\sf ApproxLocalImprove} can be executed at most $O(n^2)$ times and needs to correctly approximate $O(n^2 \poly{\log n / \ve})$ possible moves, the total number of TestMove approximations we need is $O(n^5 \poly{\log n / \ve})$. Using a union bound on Lemma \ref{lemma:testmoveapproxerror}, the approximations hold with an arbitrary constant probability (say $0.99$).
\end{proof}

%\textbf{Proof: } Note that there are $O(n)$ many calls to {\sf ApproxLocalImprove} with high probability. Since the loop on line 15 in {\sf ApproxLocalImprove} can be executed at most $O(n^2)$ times and needs to correctly approximate $O(n^2 \poly{\log n / \ve})$ possible moves, the total number of TestMove approximations we need is $O(n^5 \poly{\log n / \ve})$. Using a union bound on Lemma \eref{lemma:testmoveapproxerror}, the approximations hold with an arbitrary constant probability (say $0.99$). $\qed$\\

Being able to make one single vertex move using few samples is useful, but in order to obtain $O(\tx{poly} (\log n / \ve))$ many samples for each vertex, multiple single vertex moves need to be performed at a time. If two moves are made simultaneously, they affect each other's cost only by $O(1)$, which makes the following lemma possible.

\begin{lemma}
\label{lemma:testmoveparallel}
Let TestMove$^M(\pi, V, W, u, j)$ be the realized cost of the single vertex move $u \to j$ when all the moves in $M$ are made in parallel. Then
$$|\tx{TestMove}^M(\pi, V, W, u, j) - \tx{TestMove}(\pi, V, W, u, j)| < \frac18 d \ve / \log n$$
\end{lemma}
\begin{proof}

Doing two single vertex moves $u \to j$ and $v \to i$ (such that $u \neq v$) simultaneously affects the cost of another move by $\Theta(1)$. Therefore doing $\Theta(\ve N / \log n)$ moves simultaneously (line 22, {\sf ApproxLocalImprove}) affects the cost of a single vertex move $(u \to j)$ by $\Theta(\ve^2 N / \log^2 n)$. Since $l \geq B$, we get $d \geq \Theta(\ve N / \log n)$, which implies that the realized cost improvement of $(u \to j)$ would be changed by $O(\ve d / \log n)$.
\end{proof}

%\textbf{Proof: } Doing two single vertex moves $u \to j$ and $v \to i$ (such that $u \neq v$) simultaneously affects the cost of the other by at most $O(1)$. Therefore doing $\Theta(\ve N / \log n)$ moves simultaneously (line 24, ApproxLocalImprove) would affect the cost of a single vertex move by at most $\Theta(\ve^2 N / \log^2 n)$. Since $d > \Theta(\ve N / \log n)$, the change in cost would be $O(\ve N / \log n)$. Therefore by choosing appropriate constants we are done. $\qed$\\

The above lemmas are critical in demonstrating that only $O(\tx{poly}(\log n / \ve))$ sample sets per vertex are required. In the following, Lemma \ref{lemma:testmoveapproxerrorall} and Lemma \ref{lemma:testmoveparallel} are used as a base case to show that every move made achieves a significant cost improvement. Since the maximum feedback arc set size is $O(N^2)$, this leads to a bound on the number of moves that need to be made.

\begin{lemma}
\label{lemma:realizedimprovement}
Each move made in {\sf ApproxLocalImprove} has a realized cost improvement of $\Omega(\ve d / \log n)$.
\end{lemma}
\begin{proof}

If $u \to j$ is a move made by our algorithm then we have 
$$\tx{TestMove}_{E_{v,\alpha}}(\pi, V, W, u, j) > \frac12 \ve d / \log n$$
By Lemma \ref{lemma:testmoveapproxerrorall} and Lemma \ref{lemma:testmoveparallel} and the triangle inequality, we get 
$$|\text{TestMove}_{E_{u,\alpha}}(\pi, V, W, u, j) - \text{TestMove}^M(\pi, V, W, u, j)| < \frac14 \ve d / \log n$$ Therefore we get $\text{TestMove}^M(\pi, V, W, u, j)| > \frac14 \ve d / \log n$ as desired. 
\end{proof}

%\textbf{Proof: } If $u \to j$ is a move made by our algorithm then we have 
%$$\tx{TestMove}_{E_{v,\alpha}}(\pi, V, W, u, j) > \frac12 \ve d / \log n$$
%By Lemma \ref{lemma:testmoveapproxerrorall} and Lemma \ref{lemma:testmoveparallel} and the triangle inequality, we get $$|\text{TestMove}_{E_{u,\alpha}}(\pi, V, W, u, j) - \text{TestMove}^M(\pi, V, W, u, j)| < \frac14 \ve d / \log n$$ Therefore we get $\text{TestMove}^M(\pi, V, W, u, j)| > \frac14 \ve d / \log n$ as desired. $\qed$\\

\begin{lemma}
\label{lemma:enoughsvms}
$\Omega(\ve d / \log n)$ single vertex moves are required to create new moves $(u \to j)$ such that TestMove$(\pi, V, W, u, j) > \frac78 \ve d / \log n$.
\end{lemma}
\begin{proof}

Any single vertex move $(u \to j)$ such that TestMove$(\pi, V, W, u, j) > \frac58 \ve d / \log n$ will be included in $M_1$ and eventually be made unless its cost improvement goes below $\frac58 \ve d / \log n$. Each single vertex move changes the cost of another move by $O(1)$. Therefore in order to create a new single vertex move $u \to j$ such that TestMove$(\pi, V, W, u, j) > \frac78 \ve d / \log n$, we will have to make $\Omega(\ve d / \log n)$ moves.
\end{proof}

These demonstrate that every move made contributes to a substantial cost improvement, and yet several moves need to be made in order to create a new one of significant cost improvement. Therefore, the total cost improvement that is observed before a new single vertex move of significant cost improvement is created will be quite high, as is demonstrated below.

%\textbf{Proof: } Any single vertex move $(u \to j)$ such that TestMove$(\pi, V, W, u, j) > \frac58 \ve d / \log n$ will be included in $M_1$ and eventually be made unless its cost improvement goes below $\frac58 \ve d / \log n$. Every single vertex move changes the cost of another move by $O(1)$. Therefore in order to create a new single vertex move $u \to j$ such that TestMove$(\pi, V, W, u, j) > \frac78 \ve d / \log n$, we will have to make $\Omega(\ve d / \log n)$ moves. $\qed$\\

\begin{lemma}
\label{lemma:enoughsvms-fascost}
The total cost of the feedback arc set needs to be decreased by $\Omega(\ve^4 N^2 / \log^4 n)$ to create new strong moves $u \to j$ for which TestMove$(\pi, V, W, u, j) > \frac78 \ve d / \log n$.
\end{lemma}
\begin{proof}

By Lemma \ref{lemma:enoughsvms} and the fact that $d > \Theta(\ve N / \log n)$, we need to make $\Omega(\ve^2 N / \log^2 n)$ many single vertex moves to create such a new move. Since each single vertex move we make in the algorithm is such that its realized cost improvement is $\Omega(\ve d / \log n)$ (by Lemma \ref{lemma:realizedimprovement}), and that $d > \Theta(\ve N / \log n)$, the cost improvement due to each move is $\Omega(\ve^2 N / \log^2 n)$. The lemma follows.
\end{proof}

Thus, the cost of the feedback arc set is reduced heavily before new strong single vertex moves are created. However, since the feedback arc set cost induced by any permutation is $O(N^2)$, the algorithm does not need to check if such a new move has been created too many times, which leads us to the next two lemmas.

%\textbf{Proof: } By Lemma \ref{lemma:enoughsvms} and the fact that $d > \Theta(\ve N / \log n)$, we need to make $\Omega(\ve^2 N / \log^2 n)$ many single vertex moves to create such a new move. Since each single vertex move we make in the algorithm is such that its realized cost improvement is $\Omega(\ve d / \log n)$ (by Lemma \ref{lemma:realizedimprovement}), and that $d > \Theta(\ve N / \log n)$, the cost improvement due to each move is $\Omega(\ve^2 N / \log^2 n)$. The lemma follows. $\qed$\\

\begin{lemma}
\label{lemma:line15approxlocalimprove}
The loop on line 15 in {\sf ApproxLocalImprove} has $O(\ve^{-4}\log^4 n)$ iterations.
\end{lemma} 
\begin{proof}

The loop terminates when there is no remaining move $u \to j$ such that TestMove$(\pi, V, W, u, j) > \frac78 \ve d / \log n$. Note that even if multiple such moves are created at the end of one iteration, all of them will be made in the next iteration. Since we need to improve the cost by at least $\Omega(\ve^4 N^2 / \log^4 n)$ in order to get such a move by Lemma \ref{lemma:enoughsvms-fascost}, and the maximum cost of the feedback arc set is $\binom{N}{2}$ and the total cost cannot go below $0$, the outer loop can repeat at most $O(\ve^{-4}\log^4 n)$ times.
\end{proof}

\begin{lemma}
\label{lemma:line14approxlocalimprove}
The counter $i$ on line 14 in {\sf ApproxLocalImprove} can increase by $O(\ve^{-6}\log^6 n)$.
\end{lemma}
\begin{proof}

In {\sf ApproxLocalImprove}, the sampling on line 21 is done in such a way that every move is selected exactly once (unless the move fails the check on lines 18-20). Therefore the number of iterations required to exhaust all moves in $M_1$ is $O(\ve^{-2} \log^2 n)$, as described on line 17. Using Lemma \ref{lemma:line15approxlocalimprove}, we conclude that the total number of samples required per vertex is $O(\ve^{-6}\log^6 n)$.
\end{proof}

In this algorithm, each move that is made has positive cost improvement, as shown by Lemma \ref{lemma:realizedimprovement}. Since only moves $(u \to j)$ such that $\text{TestMove}_{E_{u,\alpha}}(\pi, V, W, u, j) > \ve d / \log n$ are executed, the terminating condition of the loop on line 15 in {\sf ApproxLocalImprove} is the same as the terminating condition in {\sf SampleAndRank}'s version of {\sf ApproxLocalImprove}. Now by the results in the query model \cite{ailon2012active}, a $(1+\ve)$ approximation is obtained and the proof of the first part of Theorem \ref{theorem:ptas} is concluded. A more detailed justification is given in the following subsection. 

\subsection{Proof of main theorem}
\label{subsection:main}
\begin{proof}[Proof of Theorem \ref{theorem:ptas}i]

As can be seen in Algorithm \ref{algorithm:snr3}: { \sf ApproxLocalImproveSAR}, the algorithm keeps making local optimizations on $\pi$ as long as there exists a single vertex move $(u \to j)$ that satisfies the following properties (where $l = \lceil \log |j - \rho_{\pi}(u)| \rceil$):
\begin{enumerate}
    \item $l \in [B, L]$
    \item $\tx{TestMove}_{E_{u,l}}(\pi, V, W, u, j) > \ve |j - \rho_{\pi}(u)| / \log n$
\end{enumerate}
Intuitively, at the end of the optimization, there are no `long' moves left which have a significant cost improvement (as estimated by the ensemble). Since the total number of moves that can be made by the algorithm is $O(n^2)$ and there are $O(n^2)$ many sample ensembles $E_{u,l}$, Ailon's algorithm needs the approximation used on line 15 of {\sf ApproxLocalImproveSAR} to be good with probability $1 - O(n^{-4})$. Therefore, Ailon proves the following lemma:

\begin{lemma}
\label{lemma:ailon}
[Ailon \cite{ailon2012active}] Any sample ensemble $\mathcal{S}$ used in {\sf ApproxLocalImproveSAR} is a good approximation with probability $1 - O(n^{-4})$. A good approximation satisfies the following two properties for all $u \in V$ and $j \in [n]$ such that $\log |j - \rho_{\pi}(u)| \geq B$ (let $l = \lceil |j - \rho_{\pi}(u)| \rceil$):
\begin{enumerate}
    \item $|\{x:(u,x) \in E_{u,l}\} \cap \{x : \rho_{\pi}(x) \in [\rho_{\pi}(u), j]\}| = \Omega(\ve^{-2} \log^2 n)$
    \item $|\tx{TestMove}_{E_{u,l}}(\pi, V, W, u, j) - \tx{TestMove}(\pi, V, W, u, j)| \leq \frac12 \ve |j - \rho_{\pi}(u)| / \log n$
\end{enumerate}
\end{lemma}

Using the second property above, line 15 of { \sf ApproxLocalImproveSAR}, and the triangle inequality, it can be seen that when the while loop on line 15 terminates, all `long' moves $(u \to j)$ left have a cost improvement of $O(\ve|\rho_{\pi}(u) - j|/\log n)$ where $u \in V$ and $j \in [n]$. Here, `long' refers to $l \in [B, L]$ where $l = \lceil |j - \rho_{\pi}(u)| \rceil$. In other words,
\begin{equation}
\label{eq:one}
    \forall u \in V \tx{ and } j \in [n] \tx{ s.t. } \lceil |j - \rho_{\pi}(u)| \rceil \in [B,L]: \tx{TestMove}(\pi, V, W, u, j) = O(\ve|\rho_{\pi}(u) - j|/\log n)
\end{equation}

Now consider any recursive call of { \sf ApproxLocalImproveSAR}, and let $\pi_V^1$ be the output of the call when restricted to the vertex set $V$. Similarly, let $\pi_V^*$ be the optimal permutation when the problem is restricted to only $V$ (using the weight submatrix $W_V$). Since $(u \to \pi_V^*(u))$ is a possible move that could be made for any $u \in V$, the result in the previous paragraph can be therefore weakened to 
\begin{equation}
\label{eq:two}
\tx{TestMove}(\pi_V^1, V, W_V, u, \pi_V^*(u)) = O(\ve|\rho_{\pi^1_V}(u) - \rho_{\pi^*_V}(u)|/\log n)
\end{equation}
where $u \in V$ is such that $|\rho_{\pi^1_V}(u) - \rho_{\pi^*_V}(u)| = \Omega(\ve N / \log n)$.\\

The above is precisely the post-condition that is used in {\sf ApproxLocalImproveSAR} to complete the proof of Lemma \ref{lemma:beta} mentioned in the main text. Since { \sf ApproxLocalImprove} also satisfies Equation \ref{eq:one} on termination due to Lemma \ref{lemma:realizedimprovement}, it also satisfies Equation \ref{eq:two}. Therefore Lemma \ref{lemma:beta} can be used for {\sf GetNearOptimalPermutation}. Let $C_X'$ be the cost of the at a leaf $X$ before it is optimized with {\sf AddApproxMFAS}, so we get 
\begin{align*}
    E[C(\pi^O, V, W)] & = E \lb \sum_{X \in \mathcal{I}} \beta_X \rb + E \lb \sum_{X \in \mathcal{L}} \alpha_X \rb\\
    & \leq (1+O(\ve))C^* - E \lb \sum_{X \in \mathcal{L}} C_X^* \rb + E \lb \sum_{X \in \mathcal{L}} \alpha_X \rb & [\tx{Lemma \ref{lemma:beta}}]\\
    & \leq (1+O(\ve))C^* + \sum_{X \in \mathcal{L}} \ve^3|V_X|^2 & [\tx{{\sf AddApproxMFAS}}]\\
    & \leq (1+O(\ve))C^* + \sum_{X \in \mathcal{L}} O(\ve C_X') & [\tx{Line 8, Algorithm \ref{algorithm:two}: {\sf Recurse}}]\\
    & \leq (1+O(\ve))C^*
\end{align*}
as desired. Finally, it is easy to see that the recursive part of the algorithm (including the single vertex moves), all can be executed in $O \lp \tx{poly}(n, \ve^{-1})\rp$ time. However, since we call {\sf AddApproxMFAS} with an error parameter of $\ve^3$, the base cases require $O(n^2 + 2^{\tx{poly}(1/\ve)} \log n)$ time (using $\beta = \ve^3$ and $\eta = 1/n^4$ in Fact \ref{fact:addapprox}). Since there are at most $n$ different base cases, the total time required is $O \left( \text{poly}(n) 2^{\text{poly}(1/\varepsilon)} \right)$ as desired.
\end{proof}

The next sections focus on the number of passes and space required in order to prove the second and third parts of Theorem \ref{theorem:ptas}.

\subsection{Number of Passes And Space}

\begin{lemma}
\label{lemma:onepassapproxlocalimprove}
A call to {\sf ApproxLocalImprove} requires $1$ pass and uses $O(N \poly(\log n / \ve))$ space.
\end{lemma}
\begin{proof}

In {\sf ApproxLocalImprove}, the creation of samples $E_{v,\alpha}$ on lines 5-13 can be done in 1 pass. Since there are $N$ possible values of $v$ , $O(\tx{poly}(\log n / \ve))$ possible values of $\alpha$, and $O(\tx{poly}(\log n / \ve))$ samples in each set, we use a total of $O(N \poly(\log n / \ve))$ space in order to store all the samples. 

Furthermore, the optimization on lines 15-28 requires no additional passes since we already have all our samples, and they do not need to be changed. The only thing that needs to be shown now is that the set $M_1$ uses at most $O^*(N)$ space. For $M_1$, it is possible that there is some vertex $u$ such that it has $\Theta(n)$ many moves $u \to j$ such that TestMove$_{E_{v,\alpha}}(\pi, V, W, u, j) > \frac12 \ve d / \log n$. However, that is not a problem since $\Theta(\ve N / \log n) < d$, which gives TestMove$_{E_{v,\alpha}}(\pi, V, W, u, j) > \Theta(\ve^2 N / \log^2) n)$. Since the maximum cost improvement we can get due to moving one single vertex around in a fixed permutation is $\Theta(N)$, the vertex will move at most $O(\log^2 n / \ve^2)$ times. Therefore instead of storing all of $M_1$ in one pass, we only need to store at most one move for each $u$ for every iteration of the loop on line 17. Therefore we need only $O(N \log n)$ space to store $M_1$.
\end{proof}

{\sf Recurse} induces a tree where each node corresponds to a call to {\sf ApproxLocalImprove}. This tree has $O(\log n)$ depth with high probability and two nodes at the same depth can be executed in the same pass. In addition, the total space used by nodes on the same layer is $\sum_{i \in L} O(N_i\poly(\log n / \ve)) = O(n\poly(\log n / \ve))$ by Lemma \ref{lemma:onepassapproxlocalimprove} as desired. Finally, note that the algorithm mentioned on line 2 in {\sf GetNearOptimalPermutation} can be executed in $1$ pass and $O(n \log n)$ space, concluding the justification of the second part of Theorem \ref{theorem:ptas}. 

The pass-space trade-off described by Theorem \ref{theorem:ptas}iii therefore remains to be shown. Given $O^*(n^{1 + 1/p})$ space, multiple passes of the original algorithm can be emulated in one pass. Partition the layers of the recursion tree into $p$ levels, where level $i$ contains nodes of size $N$ such that
$$\Theta(n^{1-\frac{i}{p}}) \leq N \leq \Theta(n^{1-\frac{i-1}{p}}).$$
The key idea is that samples $E_{v,\alpha}$ for all $v$ in nodes situated in the same level can be computed simultaneously. Once reliable samples that satisfy Lemma \ref{lemma:elementsinmove} are obtained, the proof from Lemma \ref{lemma:testmoveapproxerror} onwards follows, and {\sf ApproxLocalImprove} can be called for all nodes in the same pass. 

Note that condensing the layers of the tree into $p$ levels allows for the algorithm to use $p$ passes, so all that needs to be shown is that a sample size of $O(n^{1/p} \poly(\log n / \ve))$ per vertex is sufficient to prove Lemma \ref{lemma:elementsinmove}. Proceeding in a similar fashion to the proof of Lemma \ref{lemma:elementsinmove}, consider an arbitrary node $X_1$ in level $\eta$. Then 
$$\Theta(n^{1-\frac{\eta}{p}}) \leq |X_1| \leq \Theta(n^{1-\frac{\eta-1}{p}})$$
Now consider the highest ancestor $X_2$ of $X_1$ such that $|X_2| \leq \Theta(n^{1-\frac{\eta-1}{p}})$. This is the node for which the samples $E_{u,\alpha}$, for every $u \in X_2$ to be used for the nodes in the level $\eta$, are obtained.

Let $(u\to j)$ be the single vertex move desired in $X_1$, and $Y = |\tx{GetSample}(\pi, E_{u,\alpha}, u, j)|$. Let $Y_i$ be an indicator for which $Y_i = 1$ if and only if the $i$-th sample $(u, v_i) \in E_{u,\alpha}$ is such that $v_i$ is in the range of the single vertex move $(u \to j)$. Therefore $Y = \sum Y_i$, and 
$$\ex[Y] = \Theta \lp |X_2|^{-1} \ve^{-5} d n^{1/p} \log^5 n \rp \geq \lp \frac{|X_1|}{|X_2|} \ve^{-4} n^{1/p} \log^4 n \rp \geq \Theta \lp \ve^{-4} \log^4 n \rp$$
Using a Chernoff bound for binomial random variables yields $$\pr(Y \leq \Theta(\ve^{-4}\log^4 n)) \leq e^{-\Theta(\ve^{-4}\log^4 n)} \leq \Theta(n^{-6})$$
as desired. Note that since all elements were picked uniformly at random with repetition in $E_{u,\alpha}$, they are still uniformly random in the required range.

\subsection{Space-Time Trade-off for 1-pass Algorithms}\label{ssec:onepass}
The idea used earlier to achieve the pass-space trade-off for polynomial time algorithms was to emulate multiple passes of the algorithm in one pass by increasing the sampling rate per vertex. However, it can also be used to achieve a space-time trade-off for 1 pass algorithms: we can use our samples to emulate as many passes as possible in one pass. Once the node size becomes too small to approximate with the samples, we can solve each block by using an exponential time $(1+\ve)$ approximation algorithm similar to the one by Chakrabarti et al. \cite{chakrabarti_ghosh_mcgregor_vorotnikova_2019}. 
%which is a consequence of the $\ell_1-$sketch presented by Kane et al. \cite{kane2010exact}.

\begin{fact}
\label{fact:kane}
There is a single pass algorithm, such that 
given a stream of $\textrm{poly}(n)$ integer updates in the range $[-\textrm{poly}(n), \textrm{poly}(n)]$ to an underlying vector $\mathbf{x} \in \rs^N$, uses
$O(\ve^{-2} \log \delta^{-1} \log n)$ bits of memory and maintains a sketch $S \cdot x$ for a matrix $S$
of $d = O(\ve^{-2} \log \delta^{-1})$ dimensions. From $S \cdot x$, one can output a 
$(1 \pm \ve)$-approximation to $\|x\|_1$ with probability at least $1-\delta$. \cite{kane2010exact}
\end{fact}
%Further, each update
%to the sketch takes $O(\ve^{-2}$
%
%There is a stream-friendly $d-$dimensional $\ell_1-$sketch with accuracy $\ve$ and error $\delta$ that can handle $N^{O(1)}$ many $\pm 1$ updates to $\mathbf{x} \in \rs^N$ such that $d = O(\ve^{-2} \log \delta^{-1})$ and each update takes $O(\ve^{-2} \log \ve^{-1} \log \delta^{-1} \log N)$ time, given that all entries of $\mathbf{x}$ can be expressed in $O(\log N)$ bits.
\begin{fact}
\label{fact:chakra}
There is a single pass algorithm for the Minimum Feedback Arc Set problem on tournament graphs that uses $O(\ve^{-2}n \log^2 n)$ space and returns a $(1+\ve)-$approximation with probability at least $\frac23$. \cite{chakrabarti_ghosh_mcgregor_vorotnikova_2019}
\end{fact}

Using the above two facts, we give the following trade-off: 
\begin{theorem}
\label{theorem:2minusgamma}
There exists an algorithm which can solve the Minimum Feedback Arc Set problem on tournament graphs up to a $(1+\ve)-$approximation in a single pass with $O^*(n^{2-\gamma})$ space and $O \lp 2^{\Theta(n^\gamma \log n)} 2^{\text{poly}(1/\ve)} \rp$ time, for any $0 < \gamma < 1$.
\end{theorem}
\begin{proof}

We can divide the layers of the recursion tree into 2 parts: Part 1 consists of layers where all nodes $X$ are such that $|X| \geq \Theta(n^{\gamma})$ and Part 2 consists of the remaining layers. Note that by increasing the sampling rate to $\Theta(n^{1-\gamma} \poly(\log n / \ve))$ per vertex at the root node of the recursion tree, we can successfully emulate all the passes for the layers in Part 1. 

We also make another key observation here: let $V_1, ..., V_m$ be the top layer of Part 2, so we can divide the cost into two parts:
$$C(\pi, V, W) = C_{int} + C_{ext}$$
$$C_{int} = \sum_{i \in [m]} \sum_{u, v \in V_i: \rho_{\pi}(u) < \rho_{\pi}(v)} W(u,v)$$
$$C_{ext} = \sum_{i,j \in [m]: i < j} \sum_{(u, v) \in (V_i, V_j)} W(u,v)$$
where $C_{int}$ is the internal cost within the blocks, and $C_{ext}$ is the external cost across different blocks.

No matter what optimization we do in Part 2, $C_{ext}$ will remain the same since the relative ordering of two vertices in two different blocks will not change. Therefore we must have that $C_{ext} \leq (1+\ve)C^*$, otherwise there is no possibility that our original algorithm or Ailon's \cite{ailon2012active} algorithm would achieve a $(1+\ve)$-approximation. 

Therefore all that needs to be shown is that $C_{int} \leq (1+\ve)C^*$. Since all vertex sets $V_1, \ldots, V_m$ are disjoint, we can non-adaptively optimize them independently of each other. This is where Fact \ref{fact:kane} is used. Despite the fact that the $\ell_1-$sketch only needs $O(\ve_0^{-2} \log \delta^{-1} \log n)$ bits of space, where $\ve_0$ is the accuracy of the sketch and $\delta$ is the probability of error, their algorithm needs to brute force over all $\Theta(n!)$ possible permutations to find the best permutation, which means that $\delta = \Theta(1 / n!)$ because of the use of a union bound. However, in our algorithm we only need to union bound over $\Theta((n^{\gamma}!) \cdot n^{1-\gamma})$ permutations since we can separately (and non-adaptively) optimize all the $O(n^{1-\gamma})$ blocks of size $O(n^{\gamma})$ each. Therefore, if $C_i$ is the final cost obtained for block $V_i$, we use only $O(\ve_0^{-2} n^{\gamma})$ space in order to obtain a $C_i \leq C_i^* + \ve_0C^*$.

If we set $\ve_0 = \frac{\ve}{n^{1-\gamma}}$, we obtain the following:
\begin{itemize}
    \item Correctness: $C_{int} = \sum_{i \in [m]} C_i \leq \sum_{i \in [m]} C_i^* + \ve_0C^* \leq C^* + m\ve_0C^* \leq (1+\ve)C^*$, where $C_i^*$ is the best possible cost for block $V_i$.
    \item Space: $O(\ve^{-2} n^{2-2\gamma} n^{\gamma}) = O^*(n^{2-\gamma})$ space as desired. Note that this is the same as the space used by our sampling approach to solve Part 1. 
    \item Time: Our sampling approach uses polynomial time, so the time complexity is dominated by the brute force part of our algorithm. This takes $O((n^{\gamma})!\poly(n) 2^{\text{poly}(1/\ve)})$ time, which is the same as $2^{\Theta(n^{\gamma} \log n)} 2^{\text{poly}(1/\ve)}$, as desired.
\end{itemize}
This concludes the proof of Theorem \ref{theorem:2minusgamma}.
\end{proof}

\section{Other Streaming Problems on Tournament Graphs}\label{sec:other}

An algorithm for two more directed graph problems is presented here.
\begin{enumerate}
    \item \ttt{SCC}: the decision problem of determining whether a given digraph is strongly connected or not. A strongly connected digraph is one where there exists a directed path between any pair $u$ and $v$ of vertices. 
    \item \ttt{SCC-FIND}: the problem of finding the strongly connected components of a digraph. A strongly connected component is a subgraph which is maximally strongly connected.
\end{enumerate}

Our space lower bounds for the \ttt{SCC} problem, which are shown later in Section \ref{sec:lowerbounds}, are prohibitive for general digraphs. To obtain fast streaming algorithms, we must turn to specific types of inputs, in this case motivating the study of tournament graphs. Tournament graphs have been the topic of study in computational social choice theory \cite{brandt2016tournament} and therefore, studying their structure, such as computing the strongly connected components, is an important goal. It turns out that one can solve the \ttt{SCC} and \ttt{SCC-FIND} problems on tournament graphs in the insertion-only streaming model, given multiple passes. The appropriate algorithm consists of two phases:

\begin{enumerate}
    \item Finding a Hamiltonian Path in the graph.
    \item Partitioning the path into segments which are strongly connected components.
\end{enumerate}

Every tournament graph is known to have at least one Hamiltonian path, and Chakrabarti et al.\cite{chakrabarti_ghosh_mcgregor_vorotnikova_2019} use an algorithm called KWIKSORT to find such a Hamiltonian path in $p$ passes that uses $\tilde O(n^{1 + 1/p})$ space with high probability. Instead of emulating quicksort like KWIKSORT does, the subroutine introduced here for discovering a Hamiltonian path is inspired from mergesort, and by doing so it guarantees $\tilde O (n^{1 + 1/p})$ space in $p$ passes.

\begin{lemma}
\label{lemma:hampath}
There exists a deterministic algorithm that can find a Hamiltonian path in a tournament graph using $\tilde O (n^{1 + 1/p})$ space and $p$ passes over the data.
\end{lemma}
\begin{proof}
To see how the merging of several paths can be done efficiently, we first look at an example of simply merging two paths.

Suppose the two Hamiltonian paths are $P_1$ and $P_2$: the left half has vertices $u_1 ..., u_n$ and the right half has vertices $v_1 ..., v_m$, in that order of the Hamiltonian path on the two halves. Now the extra information we need is $f(i) = \max_{u_jv_i \in E} j$ for all $i \in [m]$. Note that without loss of generality we can assume that every $v_i$ does have an incoming edge from $P_1$, otherwise we can just take the smallest contiguous segment containing all $v_i$ which only has backward edges, and then put that path behind $u_1$. Now we can create a Hamiltonian path as follows:
\begin{enumerate}
    \item Start at $u_1$. Initialize left counter $l$ and right counter $r$ to 1.
    \item Proceed along the path $P_1$ and increment $l$ for every edge taken.
    \item Stop as soon as we get to $u_j$ such that $f(r) = j$, and then take the edge $u_jv_r$ into the other path. This edge exists by definition of $f$.
    \item Proceed along the path $P_2$ and increment $r$ for every edge taken.
    \item Stop as soon as $f(r) > l$. That means that we have not yet reached $f(r)$ on the left side, so we can just take the edge $v_{r-1}u_{l+1}$ to get back to $P_1$. Note that this edge must exist because $f(r-1) \leq l$.
    \item Repeat steps 2-5 until $l = n$ and $r = m+1$.
\end{enumerate}
Now that we have constructed the Hamiltonian path, we need to generate the values of $f$ for this new instance. This can be done in 1 pass because we can just arbitrarily pick two instances of similar size to merge together, and assign them left or right sides. We can pair up all subproblems of similar size and therefore execute this in $\log_2 n$ passes similar to mergesort. Also note that at any point we are only storing the value of $f$ for $n$ vertices, so we are only using $O(n \log n)$ space.

We can conduct this algorithm in exactly $p$ passes by using more space to combine more than two subproblems in each pass. In the first pass, we partition the vertices into groups of size $n^{1/p}$, store all the edges within each group, and at the end compute a Hamiltonian path within these groups. In the following $p - 1$ passes, we can merge the Hamiltonian paths for $n^{1/p}$ groups at once until we obtain a Hamiltonian path for the entire set of vertices. 

To merge multiple paths at once, consider $P_1, \dots, P_{n^{1/p}}$. We can iteratively merge the $j$-th path if for each $v_i \in P_j$ we store $f_k(i) = \max_{u_l v_i \in E, u_l \in P_k} j$ for $k < j$. Intuitively, we must keep track of the rightmost vertex that has an edge to $v_i$ for each of the previous merged paths. Our process can be extended then to merge these $n^{1/p}$ paths by keeping at most $n^{1/p}$ indices per vertex, leading to $\tilde{O}(n^{1 + 1/p})$ space complexity in each of the latter passes. The same is used in the first pass. 

Thus, we can obtain a Hamiltonian path with guaranteed $\tilde{O}(n^{1 + 1/p})$ space, in comparison to the high probability guarantees of KWIKSORT, and polynomial time processing per pass.
\end{proof}

\begin{theorem}
\label{theorem:scc}
Given a Hamiltonian Path on a tournament graph, the \ttt{SCC} and \ttt{SCC-FIND} problems can be solved with an additional $O(n \log n)$ space and 1 pass.
\end{theorem}
\begin{proof}
Once we have a Hamiltonian Path $v_1$ \ttt{->} $v_2$ \ttt{->} ... \ttt{->} $v_n$, the \ttt{SCC} problem can be solved as follows:
\begin{enumerate}
    \item For each vertex $v_i$, we compute the minimum $j$ such that $(v_i, v_j) \in E$. This represents the `earliest' vertex that can be reached from $v_i$. This can be done with $O(n \log n)$ space.
    \item Now for each $v_i$, we need to find the `earliest' vertex directly connected to some vertex $v_j$ where $j > i$. This is achieved by taking the suffix minima of the result of the first step, and let $f(i)$ be this suffix for $v_i$.
    \item Finally, we check if there exists a $j$ such that $f(j) = j$ and $j \neq 1$. If there does not that means that the graph is strongly connected since $v_1$ is reachable from $v_n$ by continuously traversing the path needed to go from a current vertex $i$ to $f(i)$. Otherwise there exists some $j$ such that $f(j) = j$ and $j \neq 1$, which means that we cannot reach vertex $v_k$ from $v_j$ where for any $k < j$, and therefore, it is not strongly connected.
\end{enumerate}
Note that the last step can be modified to solve the \ttt{SCC-FIND} problem as well: it is the number of $j$ such that $f(j) = j$, and the membership sets follow by the segmentation of the path by these points where $f(j) = j$.
\end{proof}

Thus, with $\tilde O(n^{1 + 1/p})$ space and $p+1$ passes, one can solve \ttt{SCC} and \ttt{SCC-FIND} for tournament graphs.

\section{Lower Bounds for Directed Graph Problems}
\label{sec:lowerbounds}

\subsection{Single pass lower bounds}

Borradaile et al. \cite{DBLP:journals/corr/BorradaileMM14} demonstrated an $\Omega(m)$ lower bound for the space complexity of detecting digraph strong connectivity for 1-pass algorithms where $m$ is the number of edges in the graph. However, this lower bound does not take into account the number of vertices in the graph and can therefore be improved for sparse graphs. We present tighter lower bounds that do take the number of vertices into account and obtain some interesting results:

\begin{theorem}
\label{theorem:scconepass}
Given a directed graph with $n$ vertices and $m$ edges, any 1-pass streaming algorithm that solves \ttt{SCC} needs at least $\Omega(m \log \frac{n^2}{m})$ space.
\end{theorem}
\begin{proof}
Our reduction from \ttt{INDEX} proceeds as follows. Let Alice's vector $x$ have length $m\log \frac{n^2}{m}$, and Bob has an index in $[m\log \frac{n^2}{m}]$. Alice constructs a bipartite graph through the two partitions $L$ and $R$ such that $|L| = |R| = n$, where every vertex $v_i \in L$ has in-degree zero and out-degree $\frac{m}{n}$ to give the graph $m$ total edges. She also partitions $R$ into $\frac{m}{n}$ blocks $\{B_j\}_{j=1}^{\frac{m}{n}}$ of size $\frac{n^2}{m}$, where the vertices on the left side have exactly one outgoing edge to a vertex in each of these blocks, as demonstrated in Figure \ref{fig:one_pass}a. To determine which vertices in these blocks, we partition $x$ into vectors $y_i$ of length $\frac{m}{n} \log \frac{n^2}{m}$ for each node in $L$, which can be interpreted as $\frac{m}{n}$ vectors, denoted $z_j$, of length $\log \frac{n^2}{m}$ that index into the blocks of size $\frac{n^2}{m}$. Such is shown in Figure \ref{fig:one_pass}b.

\begin{figure*}[ht]
    \centering
    \includegraphics[width=0.17\textwidth]{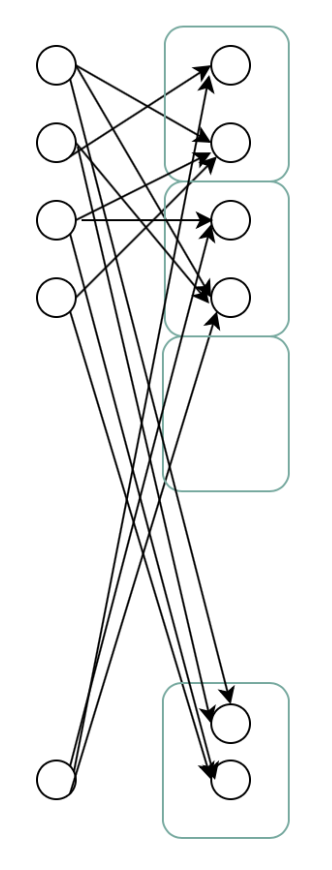}
    \includegraphics[width=0.38\textwidth]{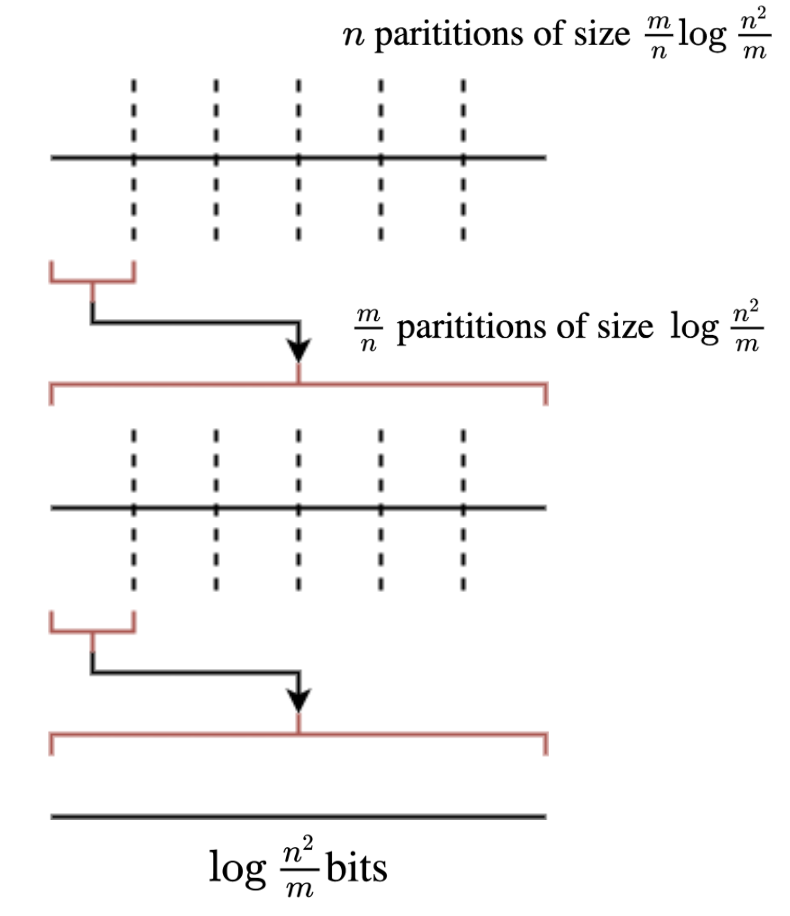}
    \includegraphics[width=0.28\textwidth]{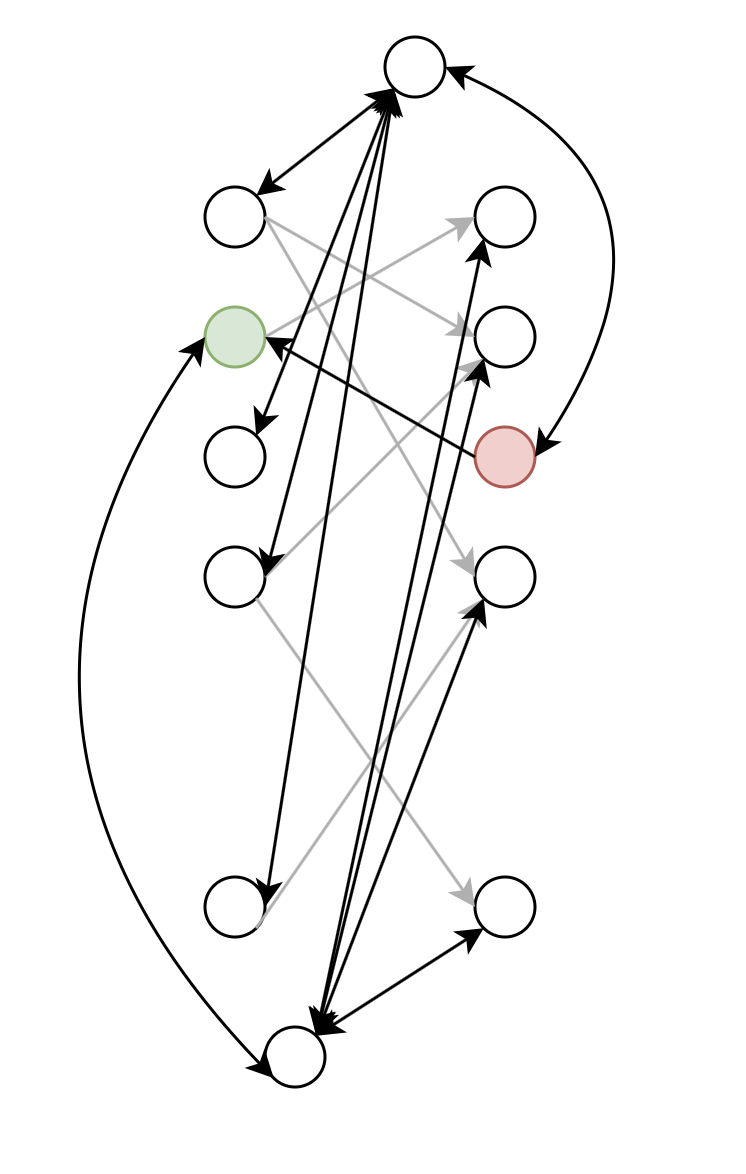}
    \captionsetup{format=hang}
    \caption{One pass reduction diagrams: on the left, Alice creates a bipartite graph by partitioning the right side into $\frac{m}{n}$ blocks and for each vertex on the left having exactly one edge to a member of each block. The middle diagram demonstrates how Alice interprets her input vector by first partitioning it into $n$ vectors each containing edge information for a vertex in $L$. On the right, Bob augments the graph input stream such that its digraph connectivity reveals the bit he is concerned with. }
    \label{fig:one_pass}
\end{figure*}

Alice then uses her bipartite graph as an input stream to the algorithm for \ttt{SCC} and sends the state of it as a message to Bob. We can also represent Bob's index $v \in [m\log \frac{n^2}{m}]$ as a $3$-tuple $(i,j,k)$, meaning Bob must determine the $k$-th bit of the vector $z_j$ that describes the outgoing edge of the $i$-th vertex in the left partition to the block $B_j$ in the right partition. Bob will append his own input stream to Alice's as follows. He first makes every vertex in $L \setminus \{ v_i\}$ strongly connected with the subset $B_j ' \subseteq B_j$ that contains exactly the elements in the block whose written index number contains a one in the $k$-th bit, which is represented by the red node in Figure \ref{fig:one_pass}c. Denote this component $S_1$. Subsequently, the input stream is augmented with edges to make $v_j$ in the left partition, represented by the green node in Figure \ref{fig:one_pass}c, strongly connected with $R \setminus B_j'$ in the right partition. Denote this component $S_2$ and note that the union of $S_1$ and $S_2$ forms the entire vertex set. Dummy nodes are employed since Bob is unaware of Alice's input stream and the pre-existence of some forward edges. In order to avoid a multigraph scenario and obey the insertion-only model's requirement that edges be unique, new nodes are created that also make the graph no longer bipartite. Finally, back edges from $B_j'$ to $v_i$ are inserted.

It suffices to demonstrate that the final graph constructed is strongly connected if and only if the edge from $v_i$ into the block $B_j$ belongs in the subset $B_j'$. Clearly with the last back edges inserted, $S_2$ is reachable from $S_1$. However, if $v_i$ has an out edge to $B_j \setminus B_j'$, then $S_1$ is not reachable from $S_2$ because the only vertex in the left partition that the set $R \setminus B_j'$ can reach is $v_i$, since Alice only appended edges going from left to right and Bob only connected these vertices to $v_i$, On the other hand, if $v_i$ has an out edge to $B_j'$, then there is a path from $S_2$ to $S_1$. Thus, Bob can ascertain the relevant bit in Alice's vector by querying the connectivity of the digraph constructed, and the lower bound on \ttt{\ttt{INDEX}} implies the desired $\Omega(m\log \frac{n^2}{m})$ lower bound on the space complexity of \ttt{SCC}.

\end{proof}

This lower bound is optimal as one can store the entire graph in this amount of memory. \\
%in the sense that, if one were to encode the entire set of edges using $\log \binom{n^2}{m}$ bits and run any off-the-shelf algorithm for finding strongly connected components, then the space complexity would be $\Omega(m \log \frac{n^2}{m})$.
 %
%Although \ttt{\ttt{INDEX}} is useful in demonstrating a strong lower bound for the single pass scenario, it cannot be extended to multiple pass algorithms that solve \ttt{SCC}. 

\begin{figure}[b]
    \centering
    \includegraphics[width=10cm, height=5cm]{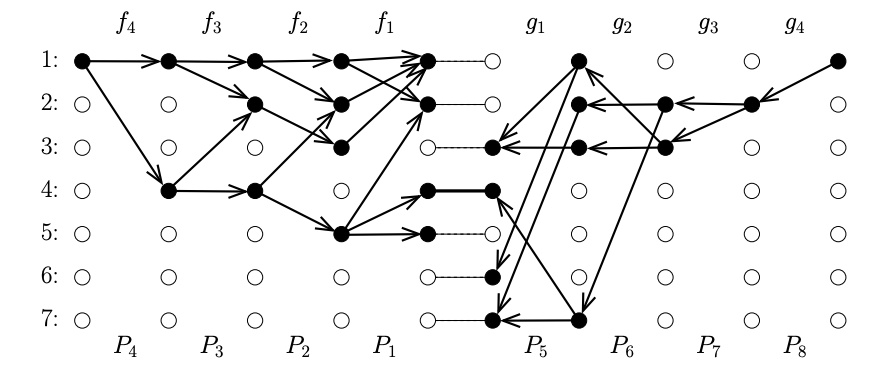}
    \caption{A true instance of \ttt{INTERSECT}(\ttt{SC}$_{7,4}$). Some edges have been omitted for clarity. The players play in order of their numbering. Example courtesy of \cite{guruswami2016superlinear}.}
    \label{fig:intersect}
\end{figure}

Similar results also follow for other directed graph problems, namely \ttt{ACYCLIC} and \ttt{S-ALL-CONN}. The former involves detecting whether or not there are any directed cycles in the input graph, while the latter involves determining whether or not there exists a path from a fixed input vertex $s$ to every single other vertex in the graph.

\begin{theorem}
\label{theorem:acyclic}
Given a directed graph with $n$ vertices and $m$ edges, any 1-pass streaming algorithm that solves \ttt{ACYCLIC} needs at least $\Omega(m \log \frac{n^2}{m})$ space.
\end{theorem}
\begin{proof}
The reduction from \ttt{INDEX} to \ttt{ACYCLIC}, which formally is the decision problem of determining whether or not a directed graph contains a cycle, follows a similar construction to our previous result for the one pass lower bound of \ttt{SCC}. Again, Alice has a $m \log \frac{n^2}{m}$ length binary vector and constructs the same bipartite graph from before. 

The augmentation to Alice's edge input stream that Bob provides however is simpler. Recall that $B_j'$ is the subset of the $j$th block that Bob needs to determine whether or not the $i$th vertex on the left side of the graph has an edge into or not. Bob now just needs to append edges going right to left, from every vertex in $B_j'$ to $v_i$, and test the whether or not the graph is acyclic. Since these are the only edges in the stream from right to left, if there is no cycle, then $v_i$ has an edge into $B_j \setminus B_j'$, else $v_i$ has an edge into $B_j'$.
\end{proof}

\begin{theorem}
\label{theorem:sall}
Given a directed graph with $n$ vertices and $m$ edges, any 1-pass streaming algorithm that solves \ttt{S-ALL-CONN} needs at least $\Omega(m \log \frac{n^2}{m})$ space.
\end{theorem}
\begin{proof}
The construction is similar to the $m$-parameterized lower bound for the \ttt{SCC} problem: Alice gets a binary input vector of length $m\log \frac{n^2}{m}$, and using the indices of the stream where the input is 1 we construct a bipartite graph where there are $n$ vertices on either side ($O(n)$ vertices in total). The right side is partitioned into $\frac{m}{n}$ blocks of size $\frac{n^2}{m}$, so it takes $O(\log \frac{n^2}{m})$ bits to encode the position of the vertex that an edge is going to in a particular block. Note that the source vertex $s$ is different from all these vertices and in Alice's construction will remain an isolated vertex.

Each index in Alice's stream can be considered as a 3-tuple $(i, j, k)$ representing an edge, where $i$ represents the node on the left of the graph, $j$ is the block number on the right side of the graph, and $k$ is the node number within that block. That edge is added if and only if the corresponding index in the stream is has the value 1. Alice also adds a dummy vertex $t$ which is also an isolated vertex.

Now Bob has an index which itself can be represented as a 3-tuple $(u, v, k)$. Bob adds the following edges to the graph:
\begin{enumerate}
    \item An edge from $s$ to $u$.
    \item Edges from the $k$-th vertex in the $v$-th block on the right side to every single other vertex in the graph (including $t$).
\end{enumerate}

Note that if edge $(u, v, k)$ was added by Alice, then $s$ can reach $u$ and the $k$-th vertex in the $v$-th block, and therefore every single vertex. However, if the edge was not added by Alice, then there is $s$ which will be able to reach $u$, and all vertices on the right side reachable from $u$, but none of them will be able to reach $t$. This completes the lower bound proof that we need at least $\Omega(m\log \frac{n^2}{m})$ space to solve this problem as well.

Note that the construction also works for the $s$-$t$-connectivity problem where we need to determine if there exists a path from $s$ to $t$ and the lower bound therefore applies to that problem as well.
\end{proof}

\subsection{Multiple pass lower bounds}

We first need to define the set chasing problem \ttt{SC}$_{n,p}$ as described by Guruswami and Onak \cite{guruswami2016superlinear} to obtain multi-pass lower bounds for \texttt{SCC}: 

\begin{definition}
\label{definition:sc}
The \ttt{SC}$_{n,p}$ problem consists of $p$ players, each with a function $f_i : [n] \to \mathcal{P}([n])$ where $\mathcal{P}(X)$ is the power set of $X$. Additionally, define $\overline{f_i} : \mathcal{P}([n]) \to \mathcal{P}([n])$ via $\overline{f_i} = \bigcup_{j \in [n]} f_i(j)$. The goal is to compute $\overline{f_1}(\overline{f_2}(\ldots \overline{f_n}(\{1\}) \ldots))$, which is output by the $p$-th player at the end of the $(p-1)$-th round, the last round.
\end{definition}

The actual problem considered for the multiple pass reduction is \ttt{INTERSECT}(\ttt{SC}$_{n,p}$), which is conveniently a decision problem.

\begin{definition}
\label{definition:intersect}
Given two instances of the \ttt{SC}$_{n,p}$, with functions $\overline{f_i}$ and $\overline{g_i}$, \ttt{INTERSECT}(\ttt{SC}$_{n,p}$) is the decision problem of checking if $\overline{f_1}(\overline{f_2}(...\overline{f_n}(\{1\})...)) \cap \overline{g_1}(\overline{g_2}(...\overline{g_n}(\{1\})...)) = \phi$.
\end{definition}

A diagrammatic representation of this problem is shown in Figure \ref{fig:intersect}. \cite{guruswami2016superlinear} uses this problem to show a multiple pass lower bound for the directed $s-t$ connectivity problem, and a similar reduction is shown here. 

\begin{theorem}
\label{theorem:sccmultiplepass}
Given a directed graph with $n$ vertices, reducing to \ttt{INTERSECT}(\ttt{SC}$_{n,p}$) demonstrates that $n^{1 + \Omega(1/p)}) / p^{O(1)}$ is a space lower bound for any $p-1$ pass algorithm that solves the \ttt{SCC} problem (and therefore the \ttt{SCC-FIND} problem).
\end{theorem}
\begin{proof}
The reduction from \ttt{SCC} to \ttt{INTERSECT}(\ttt{SC}$_{n,p}$) is constructed through the following:

\begin{enumerate}
    \item Consider the graph similar to Figure \ref{fig:intersect}, where the edges in the second half are flipped such that everything is going from left to right. This gives us a directed acyclic graph.
    \item Now we add an edge from every vertex in the last layer to $s$, and we also add an edge from every vertex in the middle red layer to $s$.
    \item We add an edge from $t$ to every vertex in the graph. Also for any vertex that does not have any outgoing neighbors we add an edge from it to $s$.
\end{enumerate}

\begin{figure}
    \centering
    \includegraphics[width=8cm,height=5cm]{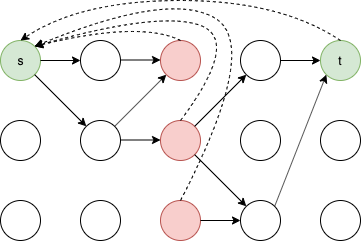}
    \caption{Reduction from \ttt{SCC} to \ttt{SC}$_{n,p}$ where $n = 3$ and $p = 2$. Quite a few edges omitted from actual construction. The middle red layer represents the output of both the \ttt{SC} problems and therefore the \ttt{INTERSECT} problem is true since there exists a vertex with edges to both sides of the field.}
    \label{fig:scctosc}
\end{figure}

We now consider the implications of \ttt{INTERSECT}(\ttt{SC}$_{n,p}$) on the original \text{SCC} instance.

\begin{itemize}
    \item If \ttt{INTERSECT}(\ttt{SC}$_{n,p}$) is false then there is no path from $s$ to $t$ since there is no vertex in the red layer that is connected to both the left and right side that is reachable from $s$. Therefore the graph cannot be strongly connected.
    \item If \ttt{INTERSECT}(\ttt{SC}$_{n,p}$) is true then we want to show that the graph is strongly connected. First note that every node is reachable from $t$ since we added an edge from $t$ to every other vertex. Now if the \ttt{INTERSECT}(\ttt{SC}$_{n,p}$) is true then $t$ is reachable from $s$ which means that every vertex is reachable from $s$. Now there are two kinds of vertices:
    \begin{itemize}
        \item Vertices that have a path to the last layer. In this case every vertex in the last layer is connected to $s$ so $s$ is reachable from the vertex.
        \item Vertices that do not have a path to the last layer. In this case they eventually hit a vertex with no out neighbors to the last layer. But in this case we added an edge from this vertex to $s$ so $s$ is reachable from the vertex.
    \end{itemize}
\end{itemize}

\end{proof}

\bibliography{ms}

\appendix

\section{Description of {\sf SampleAndRank} \cite{ailon2012active}} \label{sec:sar}

As mentioned in the main text, {\sf SampleAndRank} is similar to {\sf GetNearOptimalPermutation}, except the local optimization procedures differ:

\begin{algorithm}[H]
\caption{{\sf SampleAndRank}$(V, W, \ve)$: Top level function for computing a $(1+\ve)$-approximation to the minimum feedback arc set using only $O \lp n \tx{ poly} (\log n / \ve) \rp$ queries to $W$ in expectation.}
\begin{algorithmic}[1]
\State $n \leftarrow |V|$
\State $\pi \leftarrow $ Expected $O(1)$-approximation with expected $O(n \log n)$ queries to $W$ using Quicksort.
\State \Return {\sf RecurseSAR}$(V, W, \ve, n, \pi)$
\end{algorithmic}
\label{algorithm:snr1}
\end{algorithm}

\begin{algorithm}[H]
\caption{{\sf RecurseSAR}$(V, W, \ve, n, \pi)$: Identical recursive structure to {\sf Recurse}, except the base cases return the Trivial partition. This algorithm does not optimize those cases since its goal is not to output a near-optimal permutation, but to rather output an $\ve$-good partition. }
\begin{algorithmic}[1]
\State $N \leftarrow |V|$
\If {$N \leq \log n / \log \log n$}
\State \Return {Trivial Partition $\{ V \}$}
\EndIf
\State $E \leftarrow$ random subset of $O(\ve^{-4} \log n)$ elements from $\binom{V}{2}$ (with repetition)
\State $C \leftarrow C_E(\pi, V, W)$ \quad  ($C$ is an additive $O(\ve^2N^2)$ approximation of $C(\pi, V, W)$)
\If {$C = \Omega(\ve^2N^2)$}
\State \Return {Trivial Partition $\{ V \}$}
\EndIf
\State $\pi_1 \leftarrow$ {\sf ApproxLocalImproveSAR}$(V, W, \ve, n, \pi)$ 
\State $k \leftarrow$ uniformly random integer in the range $[N/3, 2N/3]$
\State $V_L \leftarrow \{ v \in V : \rho_{\pi}(v) \leq k \}, \pi_L \leftarrow $ restrict $\pi_1$ to $V_L$
\State $V_R \leftarrow \{ v \in V : \rho_{\pi}(v) > k \}, \pi_R \leftarrow $ restrict $\pi_1$ to $V_R$
\State \Return{Concatenate {\sf RecurseSAR}$(V_L, W, \ve, n, \pi_L)$ and {\sf RecurseSAR}$(V_R, W, \ve, n, \pi_R)$}
\end{algorithmic}
\label{algorithm:snr2}
\end{algorithm}

\begin{algorithm}[H]
\caption{{\sf ApproxLocalImproveSAR}$(V, W, \ve, n, \pi)$ \cite{ailon2012active}: The local optimization procedure for {\sf SampleAndRank}. The sampling approach differs from that of {\sf ApproxLocalImprove}. The optimization loop (line 15) has a similar terminating condition to {\sf ApproxLocalImprove}, but only one single vertex move is made at a time, which would lead to a linear number of passes in the streaming setting due to line 17, which aims to make $E_{v,i}$ independent of the resulting permutation $\pi_{u \to j}$ (we will not go into the details of the refresh step).}
\begin{algorithmic}[1]
\State $N \leftarrow |V|, B \leftarrow \lceil \log (\Theta(\ve N / \log n)) \rceil, L \leftarrow \lceil \log N \rceil$
\If {$N = O(\ve^{-3}\log^3 n)$}
\State \Return{$\pi$}
\EndIf
\For{$v \in V$}
\State $r \leftarrow \rho_{\pi}(v)$
\For {$i \in [B, L]$}
\State $E_{v,i} \leftarrow \phi$
\For {$m \in [1, \Theta(\ve^{-2} \log^2 n)]$}
\State $j \leftarrow$ integer uniformly at random chosen from $[\max\{1, r-2^i\}, \min\{n, r+2^i\}]$
\State $E_{v,i} \leftarrow E_{v,i} \cup \{(v, \pi(j)\}$
\EndFor
\EndFor
\EndFor
\While {$\exists u \in V$ and $j \in [n]$ s.t. (where $l = \lceil \log |j - \rho_{\pi}(u)| \rceil$) : $l \in [B,L]$ and $\tx{TestMove}_{E_{u,l}}(\pi, V, W, u, j) > \ve |j - \rho_{\pi}(u)| / \log n $}
\For {$v \in V$ and $i \in [B,L]$}
\State refresh sample $E_{v,i}$ with respect to the move $(u \to j)$
\EndFor
\State $\pi \leftarrow \pi_{u \to j}$
\EndWhile
\State \Return{$\pi$}
\end{algorithmic}
\label{algorithm:snr3}
\end{algorithm}

\section{{\sf AddApproxMAS} in the Semi-streaming Setting}\label{sec:addapprox-semi}

\begin{proof}[Proof of Theorem \ref{thm:addapprox}]

Frieze et al. \cite{frieze1999quick} introduce an efficient algorithm to obtain cut decompositions of matrices that are additive approximations with respect to their Frobenius norms, and demonstrate that this can be used to find a permutation that yields an additive approximation of the feedback arc set. Their algorithm can be adapted to the streaming setting, and the cut decomposition representation is highly useful since the number of cuts is only poly($1/\ve$) while each cut can be stored in $O(n)$ space. Shown below is the cut decomposition algorithm that we adapt, hereon denoted {\sf AddApproxMAS}$(V, E, \ve, \delta)$ where $V$ is the vertex set, $E$ is the edge set, $\ve$ is the accuracy parameter, and $\delta$ is the probability of correctness. The following are constants used in the pseudocode:
\begin{enumerate}
    \item $t_0 = O \lp {\ve^{-4}}\rp$
    \item $r_0 = O \lp {\ve^{-4}}\rp$
    \item $s_0 = O \lp \log \frac{t_0}{\delta} \rp$
    \item $p = O \lp \ve^{-4}\log \frac{t_0r_0s_0}{\delta} \rp$
    \item $q = O(p r_0)$
    \item $q' = O(ps_0t_0\delta^{-1} + \ve^{-8} \log \frac{s_0t_0}{\delta})$
\end{enumerate}

Frieze et al. use notation that is clarified below for a given matrix $W$, where $R, C$ represent the row and column set of the matrix $W$:

\begin{enumerate}
    \item $\mathbf{W}(S,T) = \sum_{(i,j)\in S\times T} W_{i,j}$
    \item $P_\mathbf{W}(R) = \{x\in C \mid \mathbf{W}(R, x)\geq 0 \}$
    \item $N_\mathbf{W}(R) = C\setminus P_\mathbf{W}(R)$
    \item $P_\mathbf{W}(C) = \{x\in R \mid \mathbf{W}(x, C)\geq 0 \}$
    \item $N_\mathbf{W}(C) = R\setminus P_\mathbf{W}(C)$
    \item $\text{Cut}(S,T, d)_{i,j} = \begin{cases}
    d & \text{if } (i,j)\in S\times T \\
    0 & \text{otherwise}
    \end{cases}$
    \item $G(\nu) = \begin{cases}
    \{ u \in R : \bw(u, v) \geq \nu \} & \text{if } \nu \geq 0 \\
    \{ u \in R : \bw(u, v) \leq \nu \} & \text{if } \nu < 0
    \end{cases}$
\end{enumerate}

Note that the elements of the cut decomposition are defined by $\text{Cut}(S,T, d)$ and can be efficiently stored. This is because we only need to know the row and column indices encompassed by $S$ and $T$, in addition to the value of $d$. The following is the pseudocode for obtaining the cut decomposition of an input matrix $A$, which is used as a sparse approximation of the matrix used to compute the maximum acyclic subgraph. \\

\begin{algorithm}[H]
\caption{GetCutDecomposition$(A)$ \cite{frieze1999quick}: }
\begin{algorithmic}[1]
\For{$t = 0, 1, ..., t_0 - 1$}
    \State Set $W = A - (D^{(1)} + ... + D^{(t)})$.
    \For{$s = 1, 2, ..., s_0$}
        \For{$r = 1, 2, ..., r_0$}
            \State Pick $v$ from $C$ uniformly at random.
            \State Pick $\nu$ uniformly at random from $[-1,1]$.
            \State Pick random subsets $U, U_1 \subseteq R$ independently such that $|U| = p$ and $|U_1| = q$.
            \State Pick random subset $V_1$ from $C$ independently such that $|V_1| = q$.
            \State Let $\tilde{R} \leftarrow G(\nu)$ and
            \State $\tilde{C} \leftarrow \begin{cases} 
                P_{\bw}(\tilde{R} \cap U) & \text{if } \nu > 0 \\
                N_{\bw}(\tilde{R} \cap U) & \text{if } \nu < 0
                \end{cases}$ 
            \State Compute an estimate of $\bw(\tilde{R}, \tilde{C})$ as $\tilde{W} = \frac{mn}{q^2} \bw(\tilde{R} \cap U_1, \tilde{C} \cap V_1)$.
        \EndFor
        \State Let $\tilde{R}, \tilde{C}$ be the sets which obtain the largest value of $|\tilde{W}|$ in lines 4-12.
        \State Choose new random subsets $U_1 \subseteq R, V_1 \subseteq C$ such that $|U_1| = |V_1| = q'$.
        \State Recompute $\tilde{W}$ as $\frac{mn}{q'^2} \bw(\tilde{R} \cap U_1, \tilde{C} \cap V_1)$. 
        \If{$|\tilde{W}| < \ve^2mn/9$} 
            \State Goto the next $s$ (if $s = s_0$, goto the final line 46)
        \EndIf
        \State Compute the estimate $\rho$ for $|\tilde{R}|$ as follows: $\rho = \frac{m}{q'}|\tilde{R} \cap U_1|$.
        \If{$\rho \geq 2m/5$}
            \State Goto line 30
        \Else
            \State Estimate $\bw(R, \tilde{C})$ by $W_1 = \frac{mn}{q'^2}\bw(U_1, \tilde{C} \cap V_1)$.
            \If{$W_1 \geq \ve^2mn/19$}
                \State $\tilde{R} \leftarrow R, \tilde{W} \leftarrow W_1, \rho \leftarrow m$
            \Else
                \State $\tilde{R} \leftarrow R \setminus \tilde{R}, \tilde{W} \leftarrow \frac{mn}{q'^2}\bw(U_1, V_1) - W_1, \rho \leftarrow m - \rho$
            \EndIf
        \EndIf
        \State Compute the estimate $\kappa$ for $|\tilde{C}|$ as follows: $\kappa = \frac{n}{q'}|\tilde{C} \cap V_1|$.
        \If{$\kappa \geq 2n/5$}
            \State Goto line 41
        \Else
            \State Estimate $\bw(\tilde{R}, C)$ by $W_2 = \frac{mn}{q'^2}\bw(\tilde{R} \cap U_1, V_1)$.
            \If{$W_2 \geq \ve^2mn/39$}
                \State $\tilde{C} \leftarrow C, \tilde{W} \leftarrow W_2, \kappa \leftarrow n$
            \Else
                \State $\tilde{C} \leftarrow C \setminus \tilde{C}, \tilde{W} \leftarrow \frac{mn}{q'^2}\bw(U_1, V_1) - W_2, \kappa \leftarrow n - \kappa$
            \EndIf
        \EndIf
        \State $R_{t+1} \leftarrow \tilde{R}, C_{t+1} \leftarrow \tilde{C}, d_{t+1} \leftarrow \frac{\tilde{W}}{\rho\kappa}$
        \State $D^{(t+1)} \leftarrow Cut(R_{t+1}, C_{t+1}, d_{t+1})$
        \State Goto the next $t$ (unless $t = t_0$ in which case the algorithm fails)
    \EndFor
\EndFor
\State \Return{$D^{(1)} + ... + D^{(t)}$ as the approximation to $A$}.
\end{algorithmic}
\label{algorithm:cutdecomp}
\end{algorithm}

%Although $q$ is large and the submatrices $U_1,V_1$ may be too wieldy to store, it suffices to know the indices of their row and column subsets, as they are used in {\sf AddApproxMAS}$(V, E, \ve, \delta)$ only when intersecting with $P_\bw(V')$ and $P_\bw(U')$. Their only other involvement is in the computation of $W(U_1,V_1)$ which is simple to compute in the streaming setting.

As $p$ and $q$ are poly$(1/\ve)$, then one can store each submatrix defined by the row subsets $U, U_1$ and $V, V_1$, and since $\mathbf{W}(S,T)$ is only queried for $S$ and $T$ that are respectively always subsets of $U,U_1$ or $V,V_1$, then the computations in the above algorithm can be done. As defined earlier, $t_0, s_0, r_0$ are all poly$(1/\ve)$ meaning it is feasible to keep track of the relevant information simultaneously. Once the approximation $\mathbf{D} = D^{(1)} + \dots + D^{(t)}$ for the adjacency matrix is computed, a linear program can be devised to obtain a near optimal permutation. The cost of a permutation has been shown to be related to the cost of a placement\cite{frieze1999quick, arora2002new}, where the vertices are partitioned into $X_1,\dots, X_\ell$.

$$c(\pi) = \sum_{i=1}^n\sum_{j=1}^n  W(i,j) \mathbf{1}\left[\pi(i) \in X_{i'}, \pi(j) \in X_{j'}, i' < j' \right] + \Delta $$

where $|\Delta| \leq\ve n^2/ 2$. Frieze et al. discuss solving for the best permutation under this quadratic assignment in their Section 3.3 by rewriting this cost objective in a different representation, and maximizing the objective with respect to feasibility of a linear system to ensure the existence of a permutation that satisfies the partition requirements. Thus the base case for the main algorithm is sufficient as {\sf AddApproxMAS}$(V, E, \ve, \delta)$ can be adapted to the streaming setting with suitable offline processing.

\end{proof}

\end{document}